\documentclass[10pt, oneside]{article}   	% use "amsart" instead of "article" for AMSLaTeX format
\usepackage[left=26mm, top=26mm, right=26mm, bottom=26mm]{geometry}
\usepackage{graphicx}				% Use pdf, png, jpg, or eps§ with pdflatex; use eps in DVI mode
\usepackage{amssymb,amsmath,amsfonts,bbm}
\usepackage{amsthm, thmtools,theoremref}
\usepackage[utf8]{inputenc}
\usepackage{xcolor}
\usepackage{parskip}
\usepackage{bbm}
\usepackage{verbatim}
\usepackage{hyperref}
\usepackage{csquotes}
\usepackage{appendix}

\usepackage{natbib}
\newcommand{\norm}[1]{\left\lVert #1 \right\rVert}

\usepackage{fancyhdr}
\allowdisplaybreaks

\theoremstyle{plain}
\newtheorem{thm}{Theorem}[section]
\newtheorem{lem}[thm]{Lemma}
\newtheorem{prop}[thm]{Proposition}

\theoremstyle{definition}

\newtheorem{exmp}[thm]{Example}
\newtheorem{assume}[thm]{Assumption}

\theoremstyle{remark}
\newtheorem{rem}[thm]{Remark}

\newcommand{\R}{\mathbb{R}}
\newcommand{\PP}{\mathbb{P}}
\newcommand{\Q}{\mathbb{Q}}
\newcommand{\Ex}{\mathbb{E}}

\title{Relative performance criteria of multiplicative form in complete markets}
\author{Anastasiya Tanana\thanks{Department of Mathematics, The University of Texas at Austin, \texttt{atanana@utexas.edu}.}}
\begin{document}
\maketitle

%\tableofcontents

\begin{abstract}{\footnotesize
We consider existence and uniqueness of Nash equilibria in an $N$-player game of utility maximization under relative performance criteria of multiplicative form in complete semimartingale markets. For a large class of players' utility functions, a general characterization of Nash equilibria for a given initial wealth vector is provided in terms of invertibility of a map from $\R^N$ to $\R^N$. As a consequence of the general theorem, we derive existence and uniqueness of Nash equilibria for an arbitrary initial wealth vector, as well as their convergence, if either (i)~players' utility functions are close to CRRA, or (ii)~players' competition weights are small and relative risk aversions are bounded away from infinity.
}\end{abstract}

%%%%%%%%%%%%%%%%%%%%%%

\section{Introduction}

The present paper is a study of utility maximization under relative performance criteria \emph{of multiplicative form} in complete markets. With a common finite time horizon $T$, we consider $N$ agents (e.g. fund managers) each competing with others by maximizing their expected utility of terminal wealth discounted by the geometric average of the terminal wealths of others:
$$\sup_{X^{i}_T\in\mathcal{X}(x^i_0)}\Ex\left[U_i\left(\frac{X^i_T}{(\bar{X}^i_T)^{\lambda_i}}\right)\right],\quad \text{where}\quad \bar{X}_T^i=\left(\prod_{j\neq i}X_T^j\right)^{\frac{1}{N-1}} \text{ and } \lambda_i\in[0,1],\quad i=1,...,N.$$

The model is inspired in part by \cite{lacker_mean_2019}, where the multiplicative form of relative performance criteria, both for $N$-player and mean field games, is considered for CRRA utility functions and each agent trades between a common riskless bond and an individual stock. Each individual stock is given by a log-normal process driven by two independent Brownian motions: a common market noise and an idiosyncratic noise. The model of \cite{lacker_mean_2019} thus assumes asset specialization of agents and covers as a special case the complete market case with a log-normal price process and CRRA utilities. In the current paper, we study the complete market case in a more general setting, by assuming semimartingale stock price dynamics and considering a much more general class of utility functions.

The two-player case of multiplicative relative performance criteria with asset specialization and CRRA utilities was earlier studied in \cite{basak_competition_2010}. \cite{lacker_many-player_2020} extend the model of \cite{lacker_mean_2019} to a CRRA model incorporating intertemporal consumption, \cite{anthropelos_competition_2020} and \cite{dos_reis_forward_2022} consider forward relative performance criteria.

In parallel to the above formulation, the relative performance criteria \emph{of additive form},
$$\sup_{X^{i}_T\in\mathcal{X}(x^i_0)}\Ex\left[U_i\left(X^i_T -\lambda_i \tilde{X}^i_T\right)\right],\quad \text{where}\quad \tilde{X}_T^i=\frac{1}{N-1}\sum_{j\neq i}X_T^j\ \text{ and }\ \lambda_i\in[0,1],\quad i=1,...,N,$$
 were studied extensively in mathematical literature. \cite{frei_financial_2011} and \cite{espinosa_optimal_2015} study Nash equilibria in It\^o-process market models using BSDEs in the case of exponential utility functions and incorporating portfolio constraints. See also \cite{lacker_mean_2019}, \cite{bernardin_forward_2021}, \cite{yin_n-player_2022}, \cite{bauerle_nash_2023}, and references therein.

As noted in \cite{lacker_mean_2019}, a possible interpretation of the multiplicative form of relative performance criteria is that, in competing with others, agents are comparing their \emph{rates of return} rather than their \emph{absolute wealth}, as is the case in the additive formulation.

In the present paper, for multiplicative relative performance criteria in complete markets, the main complexity comes from considering utility functions beyond CRRA. First, we characterize Nash equilibria for an $N$-player game by applying a duality theorem, Theorem \ref{thm:duality-for-complete}, to $N$ individual utility maximization problems simultaneously, this is the content of Theorem \ref{thm:NE-duality}. According to Theorem \ref{thm:NE-duality}, each Nash equilibrium for a vector $(x_0^1,...,x_0^N)\in\R^N_+$ of players' initial wealths corresponds to a vector of ``dual variables" $(C^1,...,C^N)\in\R^N_+$ and the relation is given by a system of $N$ equations and $N$ budget constraints. The system of $N$ equations becomes more tractable once we pass to a logarithmic scaling of $\R_+$. After imposing uniform lower bounds on agents' relative risk aversions (see Proposition~\ref{prop:G-invertible}, Assumption~\ref{ass:utility2}, and Remark~\ref{rem:change-of-var}), the problem of existence and uniqueness of Nash equilibria for a \emph{given} initial wealth vector $(x_0^1,...,x_0^N)\in\R^N_+$ becomes the problem of inverting a map $h:\R^N\to\R^N$, $h(D^1,...,D^N)=(\ln{x_0^1},...,\ln{x_0^N})$, from the space of ``dual variables" to the space of initial wealth vectors (both considered in logarithmic scale). This result is summarized in the main Theorem \ref{thm:NE}. In two special cases, we establish that the map $h$ is \emph{globally} invertible from $\R^N$ to $\R^N$, i.e., there exists a unique Nash equilibrium for every initial wealth vector. These two cases are as follows: 1)~each agent has a relative risk aversion that is close to being constant, 2)~each agent has a small competition weight $\lambda_i$ and a uniformly bounded away from infinity relative risk aversion. Moreover, the optimal strategy profiles in cases 1) and 2) in the limit converge, respectively, to the optimal strategy profiles corresponding to 1) CRRA utilities and 2) $N$ decoupled optimization problems with $\lambda_i=0$. In order to invert the function $h$, we apply a global inverse function theorem for $\R^N$ (Theorem~\ref{thm:global-inverse}), which states that a $C^1$ map from $\R^N$ to $\R^N$ is a diffeomorphism if and only if it is proper and its Jacobian is non-degenerate at any point. As noted by \cite{gordon_diffeomorphisms_1972}, \textit{``This theorem goes back at least to Hadamard, but it does not appear to be ``well-known". Indeed, I have found that most people do not believe it when they see it..."}

The paper is organized as follows. Section \ref{section:model} introduces the market model, the set of admissible strategy profiles, and Nash equilibria in the $N$-player game of multiplicative relative performance criteria. Section~\ref{section:NE} contains all the main results of the paper concerning the characterization of Nash equilibria. In Section~\ref{section:inv-h}, the two special cases in which $h$ is globally invertible are considered. Appendix~\ref{appendix:duality} contains the duality theorem, Theorem \ref{thm:duality-for-complete}, for individual utility maximization problems which is used in the proof of Theorem~\ref{thm:NE-duality}. Theorem \ref{thm:duality-for-complete} is similar in spirit to Theorem 2.0 of \cite{kramkov_asymptotic_1999} but more general in the sense that it incorporates a num\' eraire in the formulation of the utility maximization problem. Theorem~\ref{thm:duality-for-complete} and \cite[Theorem 2.0]{kramkov_asymptotic_1999} considered together lead to an additional observation of independent interest regarding (non-)existence of optimizers (see Remark~\ref{rem:non-existence-optimizer}). Appendix~\ref{appendix:G-inv} contains the proof of Proposition~\ref{prop:G-invertible} justifying Assumption~\ref{ass:utility2} on the utilities of the players. Appendix~\ref{appendix:proofs} contains the proofs for Section~\ref{section:inv-h}.

%%%%%%%%%%%%%%%%%%%%%%

\section{Model}\label{section:model}

We consider a complete semimartingale market model (cf. \cite{kramkov_asymptotic_1999}). There is one bond with zero interest rate and $d$ risky assets. The price process $S=(S^i)_{1\leq i\leq d}$ of the risky assets is a semimartingale on a filtered probability space $(\Omega, \mathcal{F}, (\mathcal{F}_t)_{t\in[0,T]},\PP)$, where the filtration $(\mathcal{F}_t)_{t\in[0,T]}$ satisfies the usual conditions and $T$ is a finite time horizon. 

A self-financing portfolio is a pair $(x,H)$ of initial wealth $x\in\R_+=(0,\infty)$ and predictable $S$-integrable portfolio strategy $H=(H^i)_{1\leq i\leq d}$ specifying the amount of each asset held in the portfolio. The wealth process associated with a self-financing portfolio $(x,H)$ is given by $X_t=x+\int_0^t H_udS_u$ for $t\in[0,T]$. The set of \emph{$x$-admissible self-financing wealth processes} is defined as
$$\tilde{\mathcal{X}}(x)=\left\{X_t=x+\int_0^t H_udS_u \text{ for } t\in[0,T]: X_t\geq 0\ \forall t\in[0,T],\ \PP\text{-a.s.}\right\}.$$
A probability measure $\Q\sim\PP$ is called an \textit{equivalent local martingale measure (ELMM)} if every $X\in\tilde{\mathcal{X}}(1)$ is a local martingale under $\Q$. A standard assumption, related to the absence of arbitrage opportunities in the market, is that the set $\mathcal{M}$ of equivalent local martingale measures is non-empty (see \cite{delbaen_fundamental_1998} and also Proposition~\ref{prop:primal-dom} below). In this paper, we consider a complete market case.
\begin{assume}\label{ass:mart-meas}
The market is complete, $\mathcal{M}=\{\Q\}$ is a singleton.
\end{assume}
The relative density of $\Q$ with respect to $\PP$ is denoted by $Z$ throughout the paper, $Z=\frac{d\Q}{d\PP}>0$. The expectation under the probability measure $\Q$ is denoted by $\Ex^\Q$. It can be seen that Assumption \ref{ass:mart-meas} automatically implies $\mathcal{F}=\mathcal{F}_T$.

Let us denote $\mathcal{X}(x)=\{X_T\geq0: X\in\tilde{\mathcal{X}}(x)\}$, the set of terminal wealths of all $x$-admissible self-financing wealth processes. Since every $X\in\tilde{\mathcal{X}}(x)$ is a non-negative local martingale with respect to $\Q$, $\Ex^\Q[X_T]\leq x$ for every $X_T\in\mathcal{X}(x)$. The converse of this statement is summarized in Proposition~\ref{prop:primal-dom} below, in particular, if a non-negative $\mathcal{F}$-measurable random variable $g$ satisfies $\Ex^\Q[g]=x$ then there exists a uniformly integrable $\Q$-martingale $X_t=x+\int_0^t HdS\in\tilde{\mathcal{X}}(x)$ such that $g=X_T$.

Each agent solves their individual utility maximization problem described below, with their utility function belonging to the following standard class:

\begin{assume}\label{ass:utility}
The utility function $U:(0,\infty)\to\R$ is increasing, strictly concave, continuously differentiable, and satisfies the Inada conditions $U'(0)=\lim_{x\to 0}U'(x)=\infty$ and $U'(\infty)=\lim_{x\to\infty}U'(x)=0$.
\end{assume}

It is well-known that if $U$ satisfies Assumption \ref{ass:utility} then its conjugate function $\tilde{U}$, defined by $\tilde{U}(y)=\sup_{x>0}[U(x)-xy]$ for $y>0$, is decreasing, strictly convex, continuously differentiable on $(0,\infty)$, and satisfies the Inada conditions $\tilde{U}'(0)=-\infty$ and $\tilde{U}'(\infty)=0$.

We consider utility maximization under relative performance concerns for $N\geq 2$ agents. Denote $\mathcal{N}=\{1,...,N\}$ and let $U_i$ be agent~$i$'s utility function and $x^i_0>0$ his initial wealth. We define the set of \emph{admissible strategy profiles} as the set of all $N$-tuples $(X_T^1,...,X_T^N)\in\mathcal{X}(x^1_0)\times...\times\mathcal{X}(x^N_0)$ of terminal wealths such that $X^i_T>0$ for every $i\in\mathcal{N}$. We are interested in Nash equilibria in an $N$-player game in which each agent~$i$ solves the following expected utility maximization problem
\begin{equation}\label{eq:util-max}
\sup_{X^{i}_T\in\mathcal{X}(x^i_0)}\Ex\left[U_i\left(\frac{X^i_T}{(\bar{X}^i_T)^{\lambda_i}}\right)\right],\quad i\in\mathcal{N},
\end{equation}
where $\bar{X}_T^i=\left(\prod_{j\neq i}X_T^j\right)^{\frac{1}{N-1}}$ is the geometric average of the other players' terminal wealth and $\lambda_i\in[0,1]$ is the agent's \emph{competition weight}. We define a \emph{Nash equilibrium} for \eqref{eq:util-max} to be any admissible strategy profile $(X_T^{1,*},...,X_T^{N,*})\in\mathcal{X}(x^1_0)\times...\times\mathcal{X}(x^N_0)$ such that $X_T^{i,*}$ is the optimizer for \eqref{eq:util-max} given $(X_T^{j,*})_{j\neq i}$ and $\Ex\left[U_i\left(\frac{X^{i,*}_T}{(\bar{X}^{i,*}_T)^{\lambda_i}}\right)\right]<\infty$ for every $i\in\mathcal{N}$. Hence, our definition of the Nash equilibrium \emph{excludes} the cases where at least one of the agents obtains an infinite expected utility value. For a random variable $W$ with positive and negative parts $W^+$ and $W^-$ respectively, we use the notation $\Ex[W]=\infty$ when $\Ex[W^+]=\infty$ and $\Ex[W]=-\infty$ when $\Ex[W^-]=\infty$, with a convention that
\begin{equation}\label{eq:convention}
\Ex[W^+]=\Ex[W^-]=\infty\quad \Rightarrow \quad \Ex[W]=-\infty.
\end{equation}

%%%%%%%%%%%%%%%%%%%%%%

\section{Characterization of Nash equilibria}\label{section:NE}

We now state a characterization of Nash equilibria in a complete market. By abuse of notation, when introducing vectors $(C^1,...,C^N)\in\R_+^N$, $(D^1,...,D^N)\in\R^N$, etc., we define them to be \emph{column} vectors.

\begin{thm}\label{thm:NE-duality}
Suppose that Assumption \ref{ass:mart-meas} holds and that the utility function $U_i$ satisfies Assumption~\ref{ass:utility} for every $i\in\mathcal{N}$. Let $(X^{1,*}_T,...,X^{N,*}_T)\in\mathcal{X}(x^1_0)\times...\times\mathcal{X}(x^N_0)$ be such that $X^{i,*}_T>0$ for all $i\in\mathcal{N}$. Then, the following are equivalent:
\begin{enumerate}
\item[1)] $(X^{1,*}_T,...,X^{N,*}_T)$ is a Nash equilibrium in $\mathcal{X}(x^1_0)\times...\times\mathcal{X}(x^N_0)$;
\item[2)] $\Ex^\Q[X^{i,*}_T]=x^i_0$ and there exists a vector of constants $(C^1,...,C^N)\in\R_+^N$ such that \newline $U'_i\left(\frac{X^{i,*}_T}{(\bar{X}^{i,*}_T)^{\lambda_i}}\right)=C^iZ(\bar{X}^{i,*}_T)^{\lambda_i}$ and $\Ex\left[\tilde{U}_i\left(C^iZ(\bar{X}^{i,*}_T)^{\lambda_i}\right)\right]<\infty$ for all $i\in\mathcal{N}$.
 \end{enumerate}
\end{thm}

\begin{proof}[Proof of Theorem~\ref{thm:NE-duality}]
For every $i\in\mathcal{N}$, we consider the terminal wealths $(X_T^{j,*})_{j\neq i}$ of the other players fixed and apply Theorem \ref{thm:duality-for-complete} with utility function $U_i$ and with $L=L_i=(\bar{X}^{i,*}_T)^{\lambda_i}>0$. Note that by Jensen's inequality for the concave function $x\mapsto x^{\lambda_i}$ and by H\"older's inequality we have
$$\Ex^\Q[L]\leq \left(\Ex^\Q[\bar{X}^{i,*}_T]\right)^{\lambda_i}\leq \left(\prod_{j\neq i}\Ex^\Q[X^{j,*}_T]\right)^\frac{\lambda_i}{N-1}=\left(\prod_{j\neq i}x_0^j\right)^\frac{\lambda_i}{N-1}<\infty$$
and that each optimizer obtained by Theorem \ref{thm:duality-for-complete} is automatically strictly positive and has a finite expected utility value. The joint application of Theorem \ref{thm:duality-for-complete} for individual utility maximization problems yields the stated characterization of Nash equilibria.
\end{proof}

Theorem \ref{thm:NE-duality} reformulates the problem of finding Nash equilibria for the multiplicative relative performance game in a complete market into the problem of inverting the map $F=(F^1,...,F^N):\R_+^N\to\R_+^N$ given by
\begin{equation}\label{eq:def-F}
F^i(x^1,...,x^N)=U'_i\left(\frac{x^i}{\left(\prod_{j\neq i}x^j\right)^{\frac{\lambda_i}{N-1}}}\right)\frac{1}{\left(\prod_{j\neq i}x^j\right)^{\frac{\lambda_i}{N-1}}},\quad i\in\mathcal{N}.
\end{equation}
Indeed, under the assumptions of Theorem \ref{thm:NE-duality} and up to integrability conditions,  $(X^{1,*}_T,...,X^{N,*}_T)$ is a Nash equilibrium for \emph{some} initial wealth vector if and only if there exists a vector $(C^1,...,C^N)\in\R_+^N$ such that
$$F(X^{1,*}_T,...,X^{N,*}_T)=(C^1Z,...,C^NZ).$$
It turns out that the map $F$ is invertible if we assume that each player $i$'s relative risk aversion is greater than a certain constant depending only on $\lambda_i$ and $N$.
Before proving this statement, let us transform the map $F$ into another, more digestible, one. We define a map $G=(G^1,...,G^N):\R^N\to\R^N$ as follows:
\begin{equation}\label{eq:G-from-F}
G^i(y^1,...,y^N)=\ln F^i\left(\exp{y^1},...,\exp{y^N}\right),\quad i\in\mathcal{N}.
\end{equation}
Then $F$ can be obtained back from $G$ via the relation $F^i(x^1,...,x^N)=\exp{G^i(\ln{x^1},...,\ln{x^N})}$, $i\in\mathcal{N}$,
and \textit{$F:\R_+^N\to\R_+^N$ is invertible if and only if $G:\R^N\to\R^N$ is invertible}.
Next, we define a transform of the derivative of utility function $U_i$, $i\in\mathcal{N}$:
\begin{equation}\label{eq:V-from-U}
V_i(y)=\ln{U'_i(\exp{y})}.
\end{equation} 
We recall that if a utility function $U$ is $C^2$ (twice continuously differentiable) on $(0,\infty)$ then the \textit{relative risk aversion} of $U$ is defined by $RRA[U](x)=\frac{-xU''(x)}{U'(x)}$ for $x>0$.
If $U_i:(0,\infty)\to\R$ is an increasing, strictly concave, $C^2$ utility function satisfying the Inada conditions then $V_i:\R\to\R$ is bijective, strictly decreasing, and continuously differentiable. Moreover, since 
$$-V'_i(y)=\frac{-U''_i(\exp{y})\exp{y}}{U'_i(\exp{y})}=RRA[U_i](\exp{y}),\quad y\in\R,$$
uniform upper/lower bounds on the relative risk aversion of $U_i$ are equivalent to uniform upper/lower bounds on derivative of the strictly increasing function $-V_i$.
For brevity, let us denote $\mu_i=\frac{\lambda_i}{N-1}\in[0,\frac{1}{N-1}]$. Now \eqref{eq:def-F}-\eqref{eq:V-from-U} yield that $G$ can be written in terms of $V_i$'s as follows:
\begin{equation}\label{eq:express-G}
G^i(y^1,...,y^N)=V_i\left(y^i-\sum_{j\neq i}\mu_iy^j\right)-\sum_{j\neq i}\mu_iy^j,\quad  i\in\mathcal{N}.
\end{equation}

In the next proposition we examine how subsequently adding assumptions on the relative risk aversions of players $i\in\mathcal{N}$ implies invertibility of $G$, as well as continuous differentiability and globally Lipschitz property of $G^{-1}:\R^N\to\R^N$.

\begin{prop}\label{prop:G-invertible}
Suppose that, for every $i\in\mathcal{N}$, $U_i:(0,\infty)\to\R$ is an increasing, strictly concave, $C^2$ utility function satisfying the Inada conditions $U_i'(0)=\infty$ and $U_i'(\infty)=0$. Suppose further that for every $i\in\mathcal{N}$ there exists an $\varepsilon_i>0$ such that $RRA[U_i](x)\geq \varepsilon_i+\frac{\mu_i}{1+\mu_i}$ for all $x>0$, where $\mu_i=\frac{\lambda_i}{N-1}\in[0,\frac{1}{N-1}]$. Then the map $G:\R^N\to\R^N$ defined by \eqref{eq:V-from-U} and \eqref{eq:express-G} is continuously differentiable and invertible, and hence the map $F:\R^N_+\to\R^N_+$ defined by \eqref{eq:def-F} is invertible.

Assume further that for some $i\in\mathcal{N}$ either $\lambda_i<1$, or there exists an $R_i\in(0,\infty)$ such that $RRA[U_i](x)\leq R_i$ for all $x>0$. Then the inverse $G^{-1}$ is continuously differentiable and globally Lipschitz from $\R^N$ to $\R^N$.
\end{prop}

\begin{proof}
See Appendix \ref{appendix:G-inv}.
\end{proof}

In view of Proposition~\ref{prop:G-invertible}, we upgrade Assumption \ref{ass:utility} on utility functions of the players to the following assumption, under which $G$ is continuously differentiable and invertible and $G^{-1}:\R^N\to\R^N$ is continuously differentiable and globally Lipschitz.
\begin{assume}\label{ass:utility2}
For each $i\in\mathcal{N}$, the utility function $U_i:(0,\infty)\to\R$ of agent $i$ is increasing, strictly concave, twice continuously differentiable, and satisfies the Inada conditions $U_i'(0)=\infty$ and $U_i'(\infty)=0$. With $\mu_i=\frac{\lambda_i}{N-1}\in[0,\frac{1}{N-1}]$, there exist constants $\varepsilon_i\in(0,\infty)$, $i\in\mathcal{N}$, such that 
\begin{equation}\label{eq:RRA-bound}
RRA[U_i](x)\geq \varepsilon_i+\frac{\mu_i}{1+\mu_i},\quad x>0,
\end{equation}
and for some $i\in\mathcal{N}$ either $\lambda_i<1$, or there exists an $R_i\in(0,\infty)$ such that $RRA[U_i](x)\leq R_i$ for all $x>0$.
\end{assume}
\begin{rem}\label{rem:change-of-var}
Note that $\frac{\mu_i}{1+\mu_i}\downarrow 0$ as $N\to\infty$. Furthermore, the optimization problem \eqref{eq:util-max} can be equivalently restated as
\begin{equation}\label{eq:change-of-var}
\sup_{X^{i}\in\mathcal{X}(x^i_0)}\Ex\left[\bar{U}_i\left(\frac{X^i_T}{(\bar{X}_T)^{\bar\lambda_i}}\right)\right],\quad i\in\mathcal{N},
\end{equation}
where $\bar{X}_T=\left(\prod_{j=1}^NX^j\right)^{\frac{1}{N}}$ is the geometric average of \emph{all} players' terminal wealth, $\bar{\lambda_i}=\frac{\lambda_iN}{N-1+\lambda_i}=\frac{\mu_iN}{1+\mu_i}\in[0,1]$, and $\bar{U}_i(x)=U_i(x^{1+\mu_i})$. This equivalent formulation of the problem is chosen in \cite{lacker_mean_2019} since it is better suited for passing to the mean-field limit. It can be checked that
$$RRA[\bar{U}_i](x)=-\mu_i+(1+\mu_i)RRA[U_i](x^{1+\mu_i}),\quad x>0,$$
hence the condition \eqref{eq:RRA-bound} is equivalent to $RRA[\bar{U}_i]$ being uniformly bounded away from zero (cf. Remark~3.3 in \cite{lacker_mean_2019}). Despite this simplification in the interpretation of \eqref{eq:RRA-bound} after the change of variables, we consider formulation \eqref{eq:util-max} rather than \eqref{eq:change-of-var} because it is arguably more natural for an agent in the $N$-player game to compare his performance with an \emph{average of others' performances} rather than with an \emph{average of all players' performances including himself}.
\end{rem}

To simplify the notation, we denote by the expectation, natural logarithm, or exponential function applied to a vector of (random) variables the vector of, respectively, expectations, natural logarithms, or exponential functions of the corresponding components, i.e.,
\begin{equation}\label{eq:op-notation}
\texttt{op}\ (H^1,...,H^N)=\left(\texttt{op}\ H^1,...,\texttt{op}\ H^N\right),
\end{equation}
where $\texttt{op}=\Ex,\ln,$ or $\exp$.
Theorem \ref{thm:NE-duality} and Proposition \ref{prop:G-invertible} imply that any vector $C=(C^1,...,C^N)\in\R^N_+$ produces a Nash equilibrium for \eqref{eq:util-max} given by
\begin{equation*}
\begin{aligned}
(X^{1,*}_T,...,X^{N,*}_T)&=F^{-1}(C^1Z,...,C^NZ)=\exp G^{-1}(\ln{C}+\ln{Z}\cdot\mathbbm{1})\quad \text{with}\\
(x_0^1,...,x_0^N)&=\Ex^\Q[(X^{1,*}_T,...,X^{N,*}_T)]=\Ex^\Q[\exp G^{-1}(\ln{C}+\ln{Z}\cdot\mathbbm{1})].
\end{aligned}
\end{equation*}
Similarly to passing from $F$ to $G$, we make a change of variable $D=\ln{C}\in\R^N$ and define two maps $\mathcal{H}:\R^N\mapsto(L^0_+(\Omega,\mathcal{F},\PP))^N$ and $h:\R^N\to\R^N$ as follows:
\begin{equation}\label{eq:def-of-hH}
\mathcal{H}(D)=\exp G^{-1}(D+\ln{Z}\cdot\mathbbm{1})\quad \text{and}\quad h(D)=\ln\Ex^\Q[\mathcal{H}(D)],
\end{equation}
so that each $D\in\R^N$ produces a Nash equilibrium given by
$$(X^{1,*}_T,...,X^{N,*}_T)=\mathcal{H}(D)\quad\text{with}\quad (\ln{x_0^1},...,\ln{x_0^N})=h(D).$$
Thus, in order to answer the question of existence and uniqueness of Nash equilibria for a \textit{given} initial wealth vector $(x_0^1,...,x_0^N)$, it is necessary to study the invertibility of the map $h:\R^N\to\R^N$.

For subsequent results, we further upgrade Assumption \ref{ass:mart-meas} by imposing a finiteness condition on the moments of the relative density of the unique ELMM. This condition together with the globally Lipschitz property of $G^{-1}$ are the two main technical assumptions in the proofs of the properties of $\mathcal{H}$ and $h$ that follow in the rest of the paper.
\begin{assume}\label{ass:mart-meas2}
The market is complete, $\mathcal{M}=\{\Q\}$ is a singleton. The relative density $Z=\frac{d\Q}{d\PP}$ has finite moments of all orders: $\Ex[Z^r]<\infty$ for all $r\in\R$.
\end{assume}
\begin{exmp}[Ito process market]
Let $(W_t)_{t\in[0,T]}$ be a $d$-dimensional Brownian motion on $(\Omega, \mathcal{F}, \PP)$. Assume that $(\mathcal{F}_t)_{t\in[0,T]}$ is the augmented natural filtration of $W$ and that the price of the $i$-th stock, for $i=1,...,d$, evolves according to the equation
\begin{equation*}
dS^i_t=S^i_t(b^i_tdt+\sigma^i_t\cdot dW_t),\quad S^i_0>0,\quad t\in[0,T],
\end{equation*}
where $b^i$ and $\sigma^i$ are predictable, uniformly bounded $\R$-valued and $\R^{d\times 1}$-valued processes, respectively. We assume that the volatility matrix $\sigma=(\sigma^i)_{i=1,...,d}\in\R^{d\times d}$ is invertible and that $\sigma^{-1}$ is a uniformly bounded process as well. In this case, the predictable $\R^{d\times 1}$-valued process
$\theta=(\sigma^{-1})^{T}b$ is also uniformly bounded. We denote $B_t=W_t+\int_0^t\theta_sds$. By Girsanov's theorem, $B$ is a Brownian motion under the probability measure $\Q$ equivalent to $\PP$ with relative density given by
\begin{equation}\label{eq:mart-meas}
Z=\exp\left(-\int_0^T\theta_s\cdot dW_s-\frac{1}{2}\int_0^T\norm{\theta_s}^2ds\right)=\mathcal{E}\left(-\int\theta_s\cdot dW_s \right)_T
%&=\exp\left(-\int_0^T\theta_sdB_s+\frac{1}{2}\int_0^T\norm{\theta_s}^2ds\right)
\end{equation}
and $dS^i_t=S^i_t\sigma^i_t\cdot dB_t$, $i=1,...,d$, so $\Q$ is an ELMM for $S$. It can be checked that $\Q$ is the unique ELMM for $S$ via Girsanov's theorem and that all moments of $Z$ are finite due to uniform boundedness of $\theta$ and Novikov's condition.
\end{exmp}

We can now restate Theorem \ref{thm:NE-duality} in terms of invertibility of $h$.

\begin{thm}\label{thm:NE}
Suppose that Assumptions \ref{ass:utility2} and \ref{ass:mart-meas2} hold. Let $(x_0^1,...,x_0^N)\in\R^N_+$ and let $(X^{1,*}_T,...,X^{N,*}_T)\in\mathcal{X}(x^1_0)\times...\times\mathcal{X}(x^N_0)$ be such that $X^{i,*}_T>0$ for all $i\in\mathcal{N}$. Then, the following are equivalent:
\begin{enumerate}
\item[1)] $(X^{1,*}_T,...,X^{N,*}_T)$ is a Nash equilibrium in $\mathcal{X}(x^1_0)\times...\times\mathcal{X}(x^N_0)$;
\item[2)] there exists $D\in h^{-1}(\ln{x_0^1},...,\ln{x_0^N})$ such that $(X^{1,*}_T,...,X^{N,*}_T)=\mathcal{H}(D)$.
\end{enumerate}
Thus, all Nash equilibria for \eqref{eq:util-max} with an initial wealth vector $(x^1_0,...,x^N_0)\in\R^N_+$ are given by $(X^{1,*}_{T},...,X^{N,*}_{T})=\mathcal{H}(D)$, where $D\in h^{-1}(\ln{x^1_0},...,\ln{x^N_0})$. If $(\ln{x_0^1},...,\ln{x_0^N})$ does not belong to the range of function $h$ then there is no Nash equilibrium for the initial wealth vector $(x_0^1,...,x_0^N)$.
\end{thm}

\begin{proof}
If $(X^{1,*}_T,...,X^{N,*}_T)\in\mathcal{X}(x^1_0)\times...\times\mathcal{X}(x^N_0)$ is a Nash equilibrium then by Theorem \ref{thm:NE-duality} and by definition of function $F$ there exists $(C^1,...,C^N)\in\R_+^N$ such that $F(X^{1,*}_T,...,X^{N,*}_T)=(C^1Z,...,C^NZ)$. Since $x_0^i=\Ex^\Q[X^{i,*}_T]$ for every $i\in\mathcal{N}$ by the optimality of $X^{i,*}_T\in\mathcal{X}(x^i_0)$ and since $G$ (hence $F$) is invertible under Assumption \ref{ass:utility2}, we have
$$(x_0^1,...,x_0^N)=\Ex^\Q\left[F^{-1}(C^1Z,...,C^NZ)\right]=\Ex^\Q\left[\exp{G^{-1}(\ln{C}+\ln{Z}\cdot\mathbbm{1})}\right]\overset{\eqref{eq:def-of-hH}}{=}\exp h(\ln{C}).$$
By taking $D=\ln{C}$, the implication $1)\Rightarrow 2)$ is proved. Conversely, if there exists $D\in h^{-1}(\ln{x_0^1},...,\ln{x_0^N})$ such that $(X^{1,*}_T,...,X^{N,*}_T)=\mathcal{H}(D)$ then, taking $C=\exp{D}$ and using the definitions of $h$ and $\mathcal{H}$ and Lemma~\ref{lem:integrability-satisfied}, we obtain the equivalent statement 2) in Theorem \ref{thm:NE-duality}.
\end{proof}

\begin{lem}\label{lem:integrability-satisfied}
Under Assumptions \ref{ass:utility2} and \ref{ass:mart-meas2}, for every $D\in\R^N$ and $(X^{1,*}_T,...,X^{N,*}_T)=\mathcal{H}(D)$ the integrability condition $\Ex\left[\tilde{U}_i\left(\exp(D^i)Z(\bar{X}^{i,*}_T)^{\lambda_i}\right)\right]<\infty$ is satisfied for all $i\in\mathcal{N}$.
\end{lem}

\begin{proof}
Fix $D=(D^1,...,D^N)\in\R^N$ and $i\in\mathcal{N}$. It follows from conjugacy relations between $U_i$ and $\tilde{U}_i$ that $\tilde{U}_i(y)\leq U_i(\infty)=\lim_{x\to\infty}U_i(x)$ for all $y>0$. In particular, if $U_i(\infty)\leq 0$ then the integrability condition is clearly satisfied, hence we can further assume that $U_i(\infty)> 0$. Proceeding as in the proof of Lemma~3 in \cite{kramkov_two-times_2006}, Assumption~\ref{ass:utility2} implies that for $r_i=\varepsilon_i+\frac{\mu_i}{1+\mu_i}$ and for any $a>1$ we have
$$U_i'(x)-U_i'(ax)=\int_1^a -xU_i''(tx)dt\geq r_i\int_1^a \frac{U_i'(tx)}{t}dt>r_iU_i'(ax)\ln{a},\quad z>0.$$
Hence, the inequality $U_i'(ax)<bU_i'(x)$ holds for all $x>0$, $a>1$, $b=\frac{1}{1+r_i\ln{a}}$. By Lemma~6.5 and Lemma~6.3(iii) in \cite{kramkov_asymptotic_1999}, this property implies a bound on the growth of the conjugate function $\tilde{U}_i$ near $0$: there exist $y_0>0$ and $\delta>0$ such that $\tilde{U}_i(y)<(y/y_0)^{-\delta}\tilde{U}_i(y_0)$ for all $0<y<y_0$. As a consequence, there exist constants $K_1,K_2>0$ such that
\begin{equation}\label{eq:growth-of-conjugate}
\tilde{U}_i(y)<K_1y^{-\delta}+K_2\quad \text{for all}\quad y>0.
\end{equation}
Since $\ln(X^{1,*}_T,...,X^{N,*}_T)=\ln\mathcal{H}(D)=G^{-1}(D^1+\ln{Z},...,D^N+\ln{Z})$ and the function $G^{-1}$ is globally Lipschitz by Proposition \ref{prop:G-invertible}, there exist constants $M,M_0>0$ such that
\begin{equation}\label{eq:inside-conjugate-bound}
\Big\vert D^i+\ln{Z}(\omega)+\frac{\lambda_i}{N-1}\sum_{j\neq i}\ln{X}^{i,*}_T(\omega)\Big\vert\leq M\vert \ln Z(\omega)\vert +M_0,\quad \text{for all }\omega\in\Omega.
\end{equation}
Using \eqref{eq:growth-of-conjugate}-\eqref{eq:inside-conjugate-bound}, we obtain
\begin{align*}
\Ex\left[\tilde{U}_i\left(\exp(D^i)Z(\bar{X}^{i,*}_T)^{\lambda_i}\right)\right]&\leq K_1\Ex[\exp(\delta M\vert \ln Z\vert +\delta M_0)]+K_2\\
&=K_1\exp(\delta M_0)\Ex\left[\left(Z\vee \frac{1}{Z}\right)^{\delta M}\right]+K_2,
\end{align*}
which is finite by Assumption \ref{ass:mart-meas2}.
\end{proof}

\begin{lem}\label{lem:continuity}
Let Assumptions \ref{ass:utility2} and \ref{ass:mart-meas2} hold. The map $\mathcal{H}:\R^N\to (L^0_+(\Omega,\mathcal{F},\PP))^N$ defined by \eqref{eq:def-of-hH} is continuous with respect to almost sure convergence and continuous as a map from $\R^N$ to $(L^p(\Omega,\mathcal{F},\PP))^N$ and to $(L^p(\Omega,\mathcal{F},\Q))^N$ for every $p>0$. In particular, $h(D)=\ln\Ex^\Q[\mathcal{H}(D)]$ is finite for every $D\in\R^N$.
\end{lem}

\begin{proof}
Almost sure convergence $\mathcal{H}(D')\to\mathcal{H}(D)$ as $D'\to D$ follows from continuity of $G^{-1}$. Let us fix $p>0$. Since $\ln\mathcal{H}(D)=G^{-1}(D+\ln{Z}\cdot\mathbbm{1})$ and $G^{-1}$ is globally Lipschitz, there exist finite constants $M,M_0>0$ such that
$$\vert \ln\mathcal{H}^i(D)(\omega)\vert\leq M(\vert \ln Z(\omega)\vert+\norm{D}) +M_0,\quad \text{for all } i\in\mathcal{N},\ D\in\R^N,\text{ and } \omega\in\Omega,$$
where $\norm{\cdot}$ is a norm on $\R^N$ and $\mathcal{H}^i$ is the $i$-th coordinate of the map $\mathcal{H}$. Similarly as in the proof of Lemma~\ref{lem:integrability-satisfied}, finite moments of $Z$ imply that $\mathcal{H}^i(D)^p\in L^1(\Omega,\mathcal{F},\PP)$ and $\mathcal{H}^i(D)^p\in L^1(\Omega,\mathcal{F},\Q)$ for all $i\in\mathcal{N}$. Moreover, the same bound implies that, for a fixed $D\in\R^N$ and for all $D'$ in a ball of radius $1$ around $D$, the random variables $\mathcal{H}^i(D')^p(\omega)$ are uniformly bounded by a $\PP$- and $\Q$-integrable random variable, hence $\mathcal{H}^i(D')^p\to\mathcal{H}^i(D)^p$ as $D'\to D$ in $L^1(\Omega,\mathcal{F},\PP)$ and $L^1(\Omega,\mathcal{F},\Q)$ by the dominated convergence theorem.
\begin{comment}
Let $D'\in\R^N$. Since $G^{-1}$ is globally Lipschitz, we have
\begin{align*}
\vert\mathcal{H}^i(D')^p-\mathcal{H}^i(D)^p\vert&=\mathcal{H}^i(D)^p\cdot\big\vert\left(\exp\left\{p[G^{-1}]^i(D'+\ln{Z}\cdot \mathbbm{1})-p[G^{-1}]^i(D+\ln{Z}\cdot \mathbbm{1})\right\}-1\right)\big\vert\\
&\leq \mathcal{H}^i(D)^p \left(\exp\{pM \norm{D'-D}\}-1\right), \quad \text{for all } i\in\mathcal{N}\text{ and } \omega\in\Omega,
\end{align*}
for a finite constant $M>0$ and a norm $\norm{\cdot}$ on $\R^N$. Taking the expectation and sending $\varepsilon$ to zero in $\norm{D'-D}<\varepsilon$ yields $\mathcal{H}^i(D')\to\mathcal{H}^i(D)$ in $L^p$.
\end{comment}
\end{proof}

\begin{prop}\label{prop:h-properties}
Under Assumptions \ref{ass:utility2} and \ref{ass:mart-meas2}, the map $h:\R^N\to\R^N$ is globally Lipschitz and continuously differentiable with partial derivatives at point $D\in\R^N$ given by
\begin{equation}\label{eq:h-Jacobian}
\partial_j h^i(D)=\frac{\Ex^\Q\left[\partial_j [G^{-1}]^i(D+\ln Z\cdot \mathbbm{1})\cdot \mathcal{H}^i(D)\right]}{\Ex^\Q\left[\mathcal{H}^i(D)\right]}
\end{equation}
for $i,j\in\mathcal{N}$, where $h^i$, $[G^{-1}]^i$, and $\mathcal{H}^i$ are the $i$-th coordinates of the corresponding maps.
\end{prop}

\begin{proof}
Let $i\in\mathcal{N}$ and let $L^i\in\R_+$ be a Lipschitz constant for $[G^{-1}]^i$. Then for all $D_1,D_2\in\R^N$
$$\begin{aligned}
h^i(D_1)-h^i(D_2)&=\ln\frac{\Ex^\Q\left[\mathcal{H}^i(D_1)\right]}{\Ex^\Q\left[\mathcal{H}^i(D_2)\right]}=\ln\frac{\Ex^\Q\left[\mathcal{H}^i(D_2)\exp{\left\{[G^{-1}]^i(D_1+\ln Z\cdot\mathbbm{1})-[G^{-1}]^i(D_2+\ln Z\cdot\mathbbm{1})\right\}}\right]}{\Ex^\Q\left[\mathcal{H}^i(D_2)\right]}\\
&\leq \ln\frac{\Ex^\Q\left[\mathcal{H}^i(D_2)\exp{\left\{L^i\norm{D_1-D_2}\right\}}\right]}{\Ex^\Q\left[\mathcal{H}^i(D_2)\right]}=L^i\norm{D_1-D_2}.\\
\end{aligned}$$
The same is true if $D_1$ and $D_2$ are interchanged, hence $h^i$ is $L^i$-Lipschitz for every $i\in\mathcal{N}$, and so $h$ is globally Lipschitz.

In order to prove continuous differentiability of $h$, we will show that the partial derivatives of $h$ are given by \eqref{eq:h-Jacobian} and that they are continuous in $D$. Indeed, existence and \textit{continuity} of partial derivatives of a function implies that the function is continuously differentiable.

Let us fix $i,j\in\mathcal{N}$ and $D\in\R^N$, and let $e_j\in\R^N$ denote a column vector with all entries equal to zero, except the $j$-th entry which is equal to one. By the chain rule
\begin{equation}\label{eq:h-derivative}
\partial_j \mathcal{H}^i(D+\delta e_j)(\omega)=\partial_j [G^{-1}]^i(D+\delta e_j+\ln Z(\omega)\cdot \mathbbm{1})\cdot \mathcal{H}^i(D+\delta e_j)(\omega)
\end{equation}
for every $\omega\in\Omega$, $\delta\in\R$. The Lipschitz property of $G^{-1}$ implies that for constants $L^i,L_0>0$
$$[G^{-1}]^i(D+\delta e_j+\ln Z(\omega)\cdot\mathbbm{1})\leq L^i\vert \ln Z(\omega)\vert +L_0,\quad \text{for all }\omega\in\Omega\text{ and }\vert\delta\vert\leq 1.$$
Since $\vert\partial_j [G^{-1}]^i(D+\delta e_j+\ln Z(\omega)\cdot \mathbbm{1})\vert\leq L^i$, the random variable in \eqref{eq:h-derivative} is bounded by the $\Q$-integrable random variable $L^i\exp(L_0)\left(Z\vee \frac{1}{Z}\right)^{L^i}$ uniformly in $|\delta|\leq 1$. Hence, the mean value theorem and dominated convergence theorem allow interchanging the limit and expectation in
$$\mathbb{E}\left[\lim_{\delta\to 0}\frac{\mathcal{H}^i(D+\delta e_j)-\mathcal{H}^i(D)}{\delta}\right]$$
to obtain that $\partial_j\Ex^\Q\left[\mathcal{H}^i(D)\right]=\Ex^\Q\left[\partial_j\exp{[G^{-1}]^i(D+\ln Z\cdot \mathbbm{1})}\right]$. By composing the logarithmic function and $D\mapsto\Ex^\Q\left[\mathcal{H}^i(D)\right]$, \eqref{eq:h-Jacobian} is proved. Finally, both the numerator and denominator of the right-hand side of \eqref{eq:h-Jacobian} are continuous as functions of $D\in\R^N$ by dominated convergence, because the functions inside expectations are continuous in $D$ and, similarly as above, are uniformly bounded by $\Q$-integrable random variables in a neighborhood around any $D\in\R^N$. 
\end{proof}

\begin{rem}
Note that if the Jacobian $Jh=(\partial_j h^i)_{ij}$ given by \eqref{eq:h-Jacobian} is invertible at a point $D'\in\R^N$ then, by the inverse function theorem, $h$ is locally invertible: there exist open sets $\mathcal{U},\mathcal{V}\subset\R^N$ such that $D'\in\mathcal{U}$, $h(D')\in\mathcal{V}$, and $h:\mathcal{U}\to\mathcal{V}$ and $h^{-1}:\mathcal{V}\to\mathcal{U}$ are continuously differentiable bijections. So, for all initial wealth vectors $(x^1_0,...,x^N_0)$ such that $\ln(x^1_0,...,x^N_0)\in\mathcal{V}$ there exists a (maybe non-unique) Nash equilibrium given by
$$\mathcal{H}(D),\quad \text{where } D=(h\vert_{\mathcal{U}})^{-1}\circ\ln(x^1_0,...,x^N_0).$$
Moreover, by Lemma \ref{lem:continuity} and by continuity of $(h\vert_{\mathcal{U}})^{-1}\circ\ln$, the corresponding Nash equilibria are continuous in initial wealth vector $(x^1_0,...,x^N_0)\in \exp(\mathcal{V})$ with respect to convergence almost surely and in $L^p$ for all $p>0$, where by convergence in $L^p$ we mean convergence in both $(L^p(\Omega,\mathcal{F},\PP))^N$ and $(L^p(\Omega,\mathcal{F},\Q))^N$ simultaneously.
\end{rem}

%%%%%%%%%%%%%%%%%
%%%%%%%%%%%%%%%%%

\section{Inverting function $h$}\label{section:inv-h}

By the results of the previous section, if the function $h$ has a global inverse $h^{-1}:\R^N\to\R^N$ then for every initial wealth vector $(x^1_0,...,x^N_0)\in\R^N_+$ there exists a unique Nash equilibrium for \eqref{eq:util-max} and it is given by
\begin{equation}\label{eq:solution-formula}
(X^{1,*}_T,...,X^{N,*}_T)=\mathcal{H}\circ h^{-1}\circ \ln(x^1_0,...,x^N_0).
\end{equation}
If $h^{-1}$ is continuous then by Lemma \ref{lem:continuity} the optimal strategies \eqref{eq:solution-formula} are continuous in initial wealth with respect to convergence almost surely and in $L^p$, $p>0$.

Next, we prove that $h:\R^N\to\R^N$ is a diffeomorphism in the following two cases:
\begin{enumerate}
\item The relative risk aversion of each player satisfies the lower bound \eqref{eq:RRA-bound} and is close to being constant.
\item The competition weights $\lambda_i$ of all players are close to zero and $RRA[U_i]$ is uniformly bounded away from infinity for every $i\in\mathcal{N}$.
\end{enumerate}

We shall apply a \textit{global inverse function theorem} for maps from $\R^N$ to $\R^N$. We recall that a differentiable map $f:\R^N\to\R^N$ is called a \textit{diffeomorphism} if it is bijective and its inverse is also differentiable. A continuous map $f$ is \textit{proper} if $f^{-1}(K)$ is compact whenever $K\subseteq \R^N$ is compact or, equivalently, if $\norm{x}\to\infty$ implies $\norm{f(x)}\to\infty$.

\begin{thm}\label{thm:global-inverse}
A $C^1$ map $f$ from $\R^N$ to $\R^N$ is a diffeomorphism if and only if $f$ is proper and the Jacobian $Jf$ is non-degenerate at any point.
\end{thm}

\begin{proof}
See \cite{gordon_diffeomorphisms_1972,gordon_addendum_1973} and references therein or \cite[Corollary~4.3]{palais_natural_1959}. A more elementary (or self-contained) version of the proof can be found in \cite{wu_global_1972}.
\end{proof}

In the next two subsections, we show that in the two scenarios outlined above function $h$ is proper and $Jh$ is invertible. Since $h$ is $C^1$ by Proposition \ref{prop:h-properties}, this would imply that $h$ is a global diffeomorphism of $\R^N$.

%%%%%%%%%%%%%%%%%%%%%%%%
%%%%%%%%%%%%%%%%%%%%%%%%

\subsection{Close to CRRA}

We recall that a utility function $U$ has a \textit{constant relative risk aversion} $r\geq0$ (and is said to belong to the class of \textit{CRRA} utility functions) if it is equal, up to an affine transformation, to
\begin{equation}\label{eq:CRRA-utility}
U(x)=\begin{cases}
\frac{x^{1-r}}{1-r},& r\geq0, r\neq 1,\\
\ln(x),& r=1.
\end{cases}
\end{equation}

Let $\frac{\mu_i}{1+\mu_i}<r_i<\infty$ for $i\in\mathcal{N}$ so that all of the results so far hold if player $i$ has utility function of the form \eqref{eq:CRRA-utility} with constant relative risk aversion $r_i$ for $i\in\mathcal{N}$. In this case, $V_i(y)= -r_iy$ and $G=G_c$ is a linear invertible map satisfying $G_c(0)=0$ (the subscript $c$ for variables and functions will refer to \underline{c}onstant relative risk aversion). We denote its constant invertible (by Proposition \ref{prop:G-invertible}) Jacobian matrix by
\begin{equation}\label{eq:J-CRRA}
J_c=\begin{pmatrix}
-r_1 & r_1\mu_1-\mu_1 & ... & r_1\mu_1-\mu_1\\
r_2\mu_2-\mu_2 & -r_2 & ... & r_2\mu_2-\mu_2\\
\vdots & \vdots & \ddots & \vdots \\
r_N\mu_N-\mu_N & r_N\mu_N-\mu_N & ... & -r_N
\end{pmatrix}.
\end{equation}
By \eqref{eq:h-Jacobian}, the Jacobian of $h=h_c$ is given by $Jh_c=J[G_c^{-1}]=J_c^{-1}$, constant and invertible as well, hence $h_c(D)=h_c(0)+J_c^{-1}D$ is a linear bijective map from $\R^N$ to $\R^N$. Define 
$$(A^1,...,A^N)=J_c^{-1}\mathbbm{1}\in\R^N.$$
By \eqref{eq:def-of-hH}, $h_c(0)=\ln\Ex^\Q\left[\exp(\ln{Z}\cdot A)\right]=(\ln\Ex^\Q[Z^{A^1}],...,\ln\Ex^\Q[Z^{A^N}])$. For an arbitrary given initial wealth vector $(x^1_0,...,x^N_0)\in\R^N_+$ and for utility functions $U_i$ given by \eqref{eq:CRRA-utility} with $RRA[U_i]\equiv r_i>\frac{\mu_i}{1+\mu_i}$ for all $i\in\mathcal{N}$, we obtain the unique Nash equilibrium via the formula
\begin{equation}\label{eq:CRRA-NE}
\begin{aligned}
(X^{1,*}_{T,c},...,X^{N,*}_{T,c})&\overset{\eqref{eq:solution-formula}}{=}\mathcal{H}_c\circ h_c^{-1}\circ \ln(x^1_0,...,x^N_0)\\
&=\exp\left(J_c^{-1}[ J_c\{\ln(x^1_0,...,x^N_0)-h_c(0)\}+\ln{Z}\cdot\mathbbm{1}]\right)\\
&=\exp\left(\ln(x^1_0,...,x^N_0)-h_c(0)+\ln{Z}\cdot A\right)=\left(\frac{x^1_0\cdot Z^{A^1}}{\Ex^\Q[Z^{A^1}]},...,\frac{x^N_0\cdot Z^{A^N}}{\Ex^\Q[Z^{A^N}]}\right).
\end{aligned}
\end{equation}

The invertibility of $h$ is preserved under sufficiently small fluctuations of relative risk aversion of player $i$ around $r_i$, as is summarized in the following proposition.

\begin{prop}\label{prop:close-to-CRRA}
Let Assumptions \ref{ass:utility2} and \ref{ass:mart-meas2} hold and let $\frac{\mu_i}{1+\mu_i}<r_i<\infty$. Then for $\varepsilon>0$ small enough, the condition
\begin{equation}\label{eq:almost-CRRA}
\vert RRA[U_i](x)-r_i\vert\leq \varepsilon,\quad \text{for all } x>0,\ i\in\mathcal{N},
\end{equation}
implies that $h:\R^N\to\R^N$ is a diffeomorphism. Therefore, for $N$ players solving \eqref{eq:util-max} with an arbitrary given initial wealth vector $(x_0^1,...,x_0^N)\in\R^N_+$ there exists a unique Nash equilibrium $(X^{1,*}_T,...,X^{N,*}_T)$. As~$\varepsilon\downarrow 0$, $(X^{1,*}_T,...,X^{N,*}_T)$ converges to $(X^{1,*}_{T,c},...,X^{N,*}_{T,c})$ given by \eqref{eq:CRRA-NE} almost surely and in $L^p$ for all $p>0$.
\end{prop}

\begin{proof} See Appendix \ref{appendix:proofs}.
\end{proof}

%%%%%%%%%%%%%%%%%%%%%%%%
%%%%%%%%%%%%%%%%%%%%%%%%

\subsection{Small $\lambda_i$}

Next we consider the case when players are not much concerned with competing with others, that is, $\lambda_i$'s are close to zero. In case when $\lambda_i=0$ for all $i\in\mathcal{N}$, the $N$-player game described by \eqref{eq:util-max} decouples into $N$ independent optimization problems
\begin{equation}\label{eq:zero-NE}
\sup_{X^{i}_T\in\mathcal{X}(x^i_0)}\Ex[U_i(X^i_T)],\quad i\in\mathcal{N},
\end{equation}
and each under Assumptions \ref{ass:utility2} and \ref{ass:mart-meas2} has, by Theorem \ref{thm:duality-for-complete} and by an argument analogous to the proof of Lemma \ref{lem:integrability-satisfied} (to show finiteness of the dual value function), the unique optimizer given by $X^{i,*}_{T,n}=I_i(y^iZ)$ for a constant $y^i>0$ such that $\Ex^\Q[I_i(y^iZ)]=x_0^i$, where $I_i=(U'_i)^{-1}$ is the inverse marginal utility function (the subscript $n$ for variables and functions, as in $X^{i,*}_{T,n}$, will refer to the current \underline{n}o-competition case: $\lambda_i=0$ for all $i\in\mathcal{N}$). This amounts to the corresponding functions $G_n, h_n, \mathcal{H}_n$ being diagonal: for $i\in\mathcal{N}$ and $y,D\in\R^N$, we have
\begin{equation}\label{eq:no-comp-functions}
G^i_n(y)=V_i(y^i),\quad h^i_n(D)=\ln\Ex^\Q\left[\exp{V_i^{-1}(D^i+\ln{Z})}\right],\quad \mathcal{H}_n^i(D)=\exp{V_i^{-1}(D^i+\ln{Z})}.
\end{equation}
Indeed, the optimizers given by $X^{i,*}_{T,n}=I_i(y^iZ)$ with $y^i>0$ such that $\Ex^\Q[I_i(y^iZ)]=x_0^i$ via Theorem~\ref{thm:duality-for-complete} and by 
$(X^{1,*}_{T,n},...,X^{N,*}_{T,n})=\mathcal{H}_n\circ h_n^{-1}\circ \ln(x^1_0,...,x^N_0)$
are identical, since $V_i^{-1}(z)=\ln{I_i(\exp{z})}$ and $y^i=\exp((h^i_n)^{-1}\circ\ln(x_0^i))$.

\begin{prop}\label{prop:lambda-small}
Let Assumptions \ref{ass:utility2} and \ref{ass:mart-meas2} hold and assume that for each $i\in\mathcal{N}$ there exists an $R_i>0$ such that $RRA[U_i](x)\leq R_i$ for all $x>0$. Then for $\varepsilon>0$ small enough, the condition
$$0\leq\lambda_i\leq \varepsilon,\quad \text{for all } i\in\mathcal{N},$$
implies that $h:\R^N\to\R^N$ is a diffeomorphism. Therefore, for $N$ players solving \eqref{eq:util-max} with an arbitrary given initial wealth vector $(x_0^1,...,x_0^N)\in\R^N_+$ there exists a unique Nash equilibrium $(X^{1,*}_T,...,X^{N,*}_T)$. As~$\varepsilon\downarrow 0$, $(X^{1,*}_T,...,X^{N,*}_T)$ converges to $(X^{1,*}_{T,n},...,X^{N,*}_{T,n})$, the vector of optimizers of the $N$ decoupled utility mazimization problems \eqref{eq:zero-NE}, almost surely and in $L^p$ for all $p>0$.
\end{prop}

\begin{proof} See Appendix \ref{appendix:proofs}.
\end{proof}

%%%%%%%%%%%%%%%%%%%%%%
%%%%%%%%%%%%%%%%%%%%%%

\begin{appendices}

\section{Duality in a complete market}\label{appendix:duality}

In this appendix, we state and prove a version of duality theorem for utility maximization in complete markets, tailored to our particular needs in this paper. We denote
$$\mathcal{C}(x)=\{g\in L_+^0(\Omega,\mathcal{F},\PP): g\leq X_T \text{ for some }X\in\tilde{\mathcal{X}}(x)\}.$$
We state for reference a well established characterization of sets $\mathcal{C}(x)$ via equivalent local martingale measures (see \cite[p.~926]{kramkov_asymptotic_1999}, even though the proof is omitted there) and sketch how this characterization follows from the results of \cite{delbaen_fundamental_1998}.
\begin{prop}\label{prop:primal-dom}
Assume that $\mathcal{M}\neq\emptyset$. Then for $x>0$ and $g\in L_+^0(\Omega,\mathcal{F},\PP)$,
$$g\in\mathcal{C}(x)\quad \Leftrightarrow\quad \Ex^\Q[g]\leq x\ \text{ for all }\ \Q\in\mathcal{M}.$$
In particular, if $\Q$ is the unique ELMM then $g\in\mathcal{C}(x)$ if and only if $\Ex^\Q[g]\leq x$.
\end{prop}

\begin{proof}
``$\Rightarrow$". We have $g\leq X_T$ for some $X\in\tilde{\mathcal{X}}(x)$ and, by the definition of $\mathcal{M}$, $X$ is a non-negative local martingale under every $\Q\in\mathcal{M}$. In particular,  $X$ is a supermartingale with $X_0=x$ under every $\Q\in\mathcal{M}$, implying $\Ex^\Q[g]\leq \Ex^\Q[X_T]\leq x$.

``$\Leftarrow$". Let us define $\mathcal{M}^e_s=\{\Q: \Q\sim\PP\text{ and every }X\in\tilde{\mathcal{X}}(1)\text{ is a supermartingale under }\Q\}$ and
$\mathcal{M}^e_\sigma=\{\Q:  \Q\sim\PP\text{ and } S\text{ is a sigma-martingale under }\Q\}$ (see Section~2 of \cite{delbaen_fundamental_1998} for the definition and properties of sigma-martingales). Then $\mathcal{M}\subseteq \mathcal{M}^e_s$ due to the fact that every non-negative local martingale is a supermartingale. Furthermore, by the Ansel-Stricker lemma rephrased in the context of sigma-martingales (see \cite[Theorem 5.3]{delbaen_fundamental_1998}),  if $X$ is a sigma-martingale under $\Q$ then every $X\in\tilde{\mathcal{X}}(1)$ is a local martingale under $\Q$, hence $\mathcal{M}^e_\sigma\subseteq \mathcal{M}\subseteq \mathcal{M}^e_s$. The assumption $\mathcal{M}\neq\emptyset$ implies $\mathcal{M}^e_s\neq\emptyset$, which readily implies NFLVR (the No Free Lunch with Vanishing Risk condition) by its definition. In fact, by Theorem~4.1 and Proposition~4.7 of \cite{delbaen_fundamental_1998}, not only $\mathcal{M}^e_\sigma\subseteq \mathcal{M}\subseteq \mathcal{M}^e_s$ but if NFLVR holds then $\mathcal{M}^e_\sigma\neq\emptyset$, $\mathcal{M}^e_\sigma$ is dense in $\mathcal{M}^e_s$ and, as a consequence, it is dense in $\mathcal{M}$. (In particular, $\mathcal{M}\neq\emptyset$ is equivalent to NFLVR.) Under NFLVR, by Theorem~5.12 in \cite{delbaen_fundamental_1998} and since $g\geq 0$ and $\mathcal{M}^e_\sigma$ is dense in $\mathcal{M}$, we have
$$x\geq \sup_{\Q\in\mathcal{M}}\Ex^\Q[g]=\sup_{\Q\in\mathcal{M}^e_\sigma}\Ex^\Q[g]=\inf\{\alpha:\ g\leq X_T\text{ for some }X\in\tilde{\mathcal{X}}(\alpha)\}.$$
Moreover, by Lemma~5.11 in \cite{delbaen_fundamental_1998}, the infimum on the right-hand side is a minimum, i.e., $g\in\mathcal{C}(x)$.
\end{proof}

Next, let $L\in L_+^0(\Omega,\mathcal{F},\PP)$ be such that $L>0$ and $\Ex^\Q[L]<\infty$. We consider the expected utility maximization problem described by the following value function
\begin{equation}\label{eq:value-func}
u(x)=\sup_{g\in\mathcal{C}(x)}\Ex[U(g/L)],\quad x>0.
\end{equation}
The random variable $L$ can be thought of as a num\'eraire. Since $L$ is $\Q$-integrable, for every $x>0$ and $\alpha=x/\Ex^\Q[L]>0$ the random variable $\alpha L$ belongs to $\mathcal{C}(x)$. Hence, $u(x)\geq \Ex[U(\alpha L/L)]=U(\alpha)>-\infty$ for every $x>0$. In particular, if $g\in\mathcal{C}(x)$ is the optimizer for \eqref{eq:value-func} then this automatically assumes $\Ex[U^-(g/L)]<\infty$ by our convention 
\eqref{eq:convention}.

\begin{thm}\label{thm:duality-for-complete}
Let Assumptions \ref{ass:mart-meas}, \ref{ass:utility} hold, $L\in L_+^0(\Omega,\mathcal{F},\PP)$ be such that $L>0$ and $\Ex^\Q[L]<\infty$, and $g\in\mathcal{C}(x)$. Then, the following are equivalent:
\begin{enumerate}
\item[1)] $g\in\mathcal{C}(x)$ is the optimizer for \eqref{eq:value-func} and $\Ex[U(g/L)]<\infty$;
\item[2)] $\Ex^\Q[g]=x$ and there exists a constant $y>0$ such that $U'(g/L)=yLZ$, $\PP-$almost surely, and $\Ex[\tilde{U}^+(yLZ)]<\infty$.
\end{enumerate}
\end{thm}

\begin{proof} $2)\Rightarrow 1).$ Recall that, by the definition of $\tilde{U}$, we have $U(z)\leq\tilde{U}(y)+yz$ for all $y,z>0$. Hence, $\Ex[\tilde{U}^+(yLZ)]<\infty$ and $\Ex[Zg]=\Ex^\Q[g]=x$ imply $\Ex[U^+(g/L)]\leq \Ex[\tilde{U}^+(yLZ)+yLZ(g/L)]= \Ex[\tilde{U}^+(yLZ)]+y\Ex[Zg]<\infty$. Take an arbitrary $\tilde{g}\in\mathcal{C}(x)$. By a standard argument, using concavity of $U$ we obtain
\begin{align*}
\Ex[U(\tilde{g}/L)]&=\Ex[U(g/L)]+\Ex[U(\tilde{g}/L)-U(g/L)]\leq \Ex[U(g/L)]+\Ex[U'(g/L)(\tilde{g}/L-g/L)]\\
&=\Ex[U(g/L)]+\Ex[yLZ(\tilde{g}/L-g/L)]=\Ex[U(g/L)]+y\Ex[Z(\tilde{g}-g)]\leq\Ex[U(g/L)],
\end{align*}
where the last inequality follows from $\Ex[Z\tilde{g}]\leq x=\Ex[Zg]$, and equality in the first line holds only if $\PP-$almost surely $\tilde{g}=g$, by the strict concavity of $U$. This shows that $g$ is the unique optimizer in $\mathcal{C}(x)$.

$1)\Rightarrow 2).$ Since $\tilde{U}(y)\geq U(z)-yz$ for all $y,z>0$, for $\alpha=x/\Ex^\Q[L]$ we obtain $\Ex[-\tilde{U}^-(yLZ)]\geq \Ex[-U^-(\alpha) -\alpha yLZ]=-U^-(\alpha)-xy>-\infty,$
i.e., $\tilde{U}^-(yLZ)$ is integrable for every $y>0$.
 Let us define (the dual value function)
$$v(y)=\Ex[\tilde{U}(yLZ)]\in(-\infty,\infty],\quad y>0.$$
Since $u>-\infty$ on $(0,\infty)$, $u(x)= \Ex[U(g/L)]<\infty$ and $u$ is concave, $u$ is finite on $(0,\infty)$. In particular, the function $-u(-\cdot)$, where $u$ is extended to take value $-\infty$ on $(-\infty,0]$, is a proper convex function on $\R$ (here we use a different definition of a proper function than the one used earlier for maps from $\R^N$ to $\R^N$; a convex function is called proper if it never takes value $-\infty$ and is not identically equal to $\infty$). The function $v$ is non-increasing and convex on $(0,\infty)$. Next we verify that $v$, extended by the value $\infty$ on $(-\infty,0]$, is the conjugate function of $u$ (i.e., the convex conjugate of $-u(-\cdot)$, see for example \cite[p.~104]{rockafellar}), that is, $v(y)=\sup_{z>0}[u(z)-yz]$ for all $y>0$. For arbitrary fixed $y,z>0$, we have
\begin{equation}\label{eq:conjugacy}
\begin{aligned}
\sup_{z>0}[u(z)-yz]&=\sup_{z>0,\ \tilde{g}\in\mathcal{C}(z)}\Ex[U(\tilde{g}/L)-yz]=\sup_{z>0,\ \tilde{g}\in\mathcal{C}(z)}\Ex[U(\tilde{g}/L)-(yLZ)\tilde{g}/L]\\
&\leq\sup_{\tilde{g}\geq 0}\Ex[U(\tilde{g}/L)-(yLZ)\tilde{g}/L]=\Ex[\tilde{U}(yLZ)]=v(y).
\end{aligned}
\end{equation}
Define $\tilde{U}^n(y)=\sup_{z\in(0,n)}[U(z)-yz]$, for $y>0$, $n\in\mathbb{N}$. Then since $\tilde{g}\wedge n\in\mathcal{C}(z)$ ($\tilde{g}\wedge n$ is the pointwise minimum of $\tilde{g}$ and $n$) for every $n$, every random variable $\tilde{g}\geq 0$, and for $z$ large enough, we obtain
$$\sup_{z>0,\ \tilde{g}\in\mathcal{C}(z)}\Ex[U(\tilde{g}/L)-(yLZ)\tilde{g}/L]\geq\sup_{\tilde{g}\geq 0}\Ex[U((\tilde{g}\wedge n)/L)-(yLZ)(\tilde{g}\wedge n)/L]=\Ex[\tilde{U}^n(yLZ)]=v^n(y),$$
where the last equality is the definition of $v^n$. Thus, in order to prove equality in \eqref{eq:conjugacy}, it remains to show that $v^n(y)\to v(y)$ for every $y>0$ as $n\to\infty$. Since $\tilde{U}^n\leq\tilde{U}^{n+1}\leq\tilde{U}$ for every $n\in\mathbb{N}$, we have $v^n\leq v^{n+1}\leq v$. On the other hand, for every $n\geq 2$, $\tilde{U}^n\geq\tilde{U}^2=\tilde{U}$ on $[U'(2),\infty)$ and $\tilde{U}^n\geq\tilde{U}^2=U(2)-2U'(2)$ on $(0,U'(2))$. Since $U$ is concave, $U'(2)<U(2)-U(1)$ and, summing up,
$$\tilde{U}^n\geq \tilde{U}\mathbbm{1}_{[U'(2),\infty)}+(2U(1)-U(2))\mathbbm{1}_{(0,U'(2))},\quad n\geq2.$$
Using this lower bound (uniform in $n$), integrability of $\tilde{U}^-(yLZ)$, and monotone convergence for $\tilde{U}^n(yLZ)\uparrow\tilde{U}(yLZ)$, we obtain $\lim_{n\to\infty}v^n(y)=v(y)$ for every $y>0$ and equality in place of inequality in \eqref{eq:conjugacy}.

By Theorem~12.2 in \cite{rockafellar}, $v$ is a \emph{closed proper} convex function. Taking into account that $v$ is non-increasing, closedness and properness imply that $v$ is right-continuous, hence for $y_0=\inf\{y: v(y)<\infty\}\geq 0$ the function $v$ is continuous on $[y_0,\infty)$, possibly with $v(y_0)=\infty$. By the same theorem, the bidual relation holds: $u(z)=\inf_{y>0}[v(y)+yz]$ for all $z>0$. In particular,
$u(x)=\inf_{y>0}[v(y)+xy]$. Note that the function $y\mapsto v(y)+xy$ is convex, continuous on $[y_0,\infty)$, and tends to $\infty$ as $y\to\infty$. Indeed, the Inada condition $\tilde{U}'(\infty)=0$ implies that for every $\varepsilon>0$ there exists a constant $C(\varepsilon)$ such that $\tilde{U}(y)>-\varepsilon y+C(\varepsilon)$ for all $y>0$. Hence
$$v(y)+xy=\Ex[\tilde{U}(yLZ)]+xy\geq \Ex[-\varepsilon (yLZ)+C(\varepsilon)]+xy=C(\varepsilon)+y(x-\varepsilon\Ex[LZ])\nearrow\infty$$ 
as $y\uparrow\infty$ for $\varepsilon>0$ small enough.
These properties imply that the finite infimum $u(x)$ of $y\mapsto v(y)+xy$ is attained at some $y\in[y_0,\infty)$ and, by the optimality of $g$, we obtain
$$\Ex[U(g/L)]=u(x)=v(y)+xy=\Ex[\tilde{U}(yLZ)+(yLZ)g/L]\geq \Ex[U(g/L)].$$
Equality may hold in place of inequality only if $U'(g/L)=yLZ$, $\PP-$almost surely. Note that $y>0$ since otherwise, by this formula, $g=\infty$ and $x=\Ex[Zg]=\infty$, a contradiction. The integrability of $\tilde{U}^+(yLZ)$ follows from finiteness of $u(x)$, hence $v(y)$, completing the proof.
\begin{comment}By the same theorem, the function $u^{**}(z)=\inf_{y>0}[v(y)+yz]$, $z\in\R$, is precisely the smallest upper-semicontinuous function satisfying $u^{**}\geq u$ (see \cite[p.~52]{rockafellar} for the definition of the closure of a convex function). Since $u$ is non-decreasing, $u^{**}$ is the right-continuous version of $u$. In particular, $u^{**}$ and $u$ agree on $(x_0,\infty)$, where $x_0=\inf\{z>0: u(z)>-\infty\}$.\end{comment}
 \end{proof}

\begin{rem}\label{rem:non-existence-optimizer}
In case $L\equiv 1$, Theorem \ref{thm:duality-for-complete} can be compared to Theorem~2.0 in \cite{kramkov_asymptotic_1999}. The latter states that (i) the dual value function $v$ is $C^1$ and strictly convex on $(y_0,\infty)$ with $v'(\infty)=0$ and $-v'(y)=\Ex^\Q[I(yZ)]$ for $I=(U')^{-1}$ and all $y>y_0$; (ii) for $x_0=\lim_{y\downarrow y_0}\Ex^\Q[I(yZ)]$ the primal optimizer exists for all $x<x_0$ and is given by $I(yZ)$, where $y>y_0$ is such that $x=\Ex^\Q[I(yZ)]$. By Theorem \ref{thm:duality-for-complete}, on the other hand, an optimizer exists for $x>0$ if and only if there exists $y>0$ such that $v(y)<\infty$ and $x=\Ex^\Q[I(yZ)]$, in which case the optimizer is given by $I(yZ)$. It follows from the two theorems combined that an optimizer for $u(x)$ exists \emph{only} for (i) $x<x_0$, (ii) $x=x_0=\Ex^\Q[I(y_0Z)]$ if $x_0<\infty$ and $v(y_0)<\infty$, in which case it is given by $I(y_0Z)$. In particular, there is no optimizer for $x>x_0$.
\end{rem}
%%%%%%%%%%%%%%%%%%%%%%

\section{Invertibility of $G$}\label{appendix:G-inv}

\begin{proof}[Proof of Proposition \ref{prop:G-invertible}]
The map $G$ is continuously differentiable because $V_i\in C^1$. By the lower bound on the relative risk aversion of $U_i$, $-V'_i(y)\geq\varepsilon_i+\frac{\mu_i}{1+\mu_i}$ for every $y\in\R$. To show that $G:\R^N\to\R^N$ (and therefore $F:\R^N_+\to\R^N_+$) is invertible, we will show that for every $z=(z^1,...,z^N)\in\R^N$ there exists a unique $y=(y^1,...,y^N)\in\R^N$ such that
\begin{equation}\label{eq:invert-G1}
\begin{aligned}
z^i&=V_i\left(y^i-\sum_{j\neq i}\mu_iy^j\right)-\sum_{j\neq i}\mu_iy^j\\
&=\left[V_i\left(y^i-\sum_{j\neq i}\mu_iy^j\right)+\frac{\mu_i}{1+\mu_i}\left(y^i-\sum_{j\neq i}\mu_iy^j\right)\right]-\frac{\mu_i}{1+\mu_i}\sum_{j=1}^N y^j,\quad  i\in\mathcal{N}.
\end{aligned}
\end{equation}
The function $u\mapsto V_i(u)+\frac{\mu_iu}{1+\mu_i}$ is strictly decreasing, bijective from $\R$ to $\R$, and has a strictly decreasing bijective inverse $W_i:\R\to\R$ satisfying $0>W'_i\geq -\frac{1}{\varepsilon_i}$, for every $i\in\mathcal{N}$. To solve \eqref{eq:invert-G1} for $y$, let us denote $s=\sum_{j=1}^N y^j$ and rewrite \eqref{eq:invert-G1} as
\begin{equation}\label{eq:invert-G2}
W_i\left(z^i+\frac{\mu_i}{1+\mu_i}s\right)=y^i-\sum_{j\neq i}\mu_iy^j=(1+\mu_i)y^i-\mu_is,\quad i\in\mathcal{N}.
\end{equation}
Dividing both sides by $1+\mu_i$ and summing over $i$, we obtain
\begin{equation}\label{eq:invert-G3}
\sum_{i=1}^N\frac{1}{1+\mu_i}W_i\left(z^i+\frac{\mu_i}{1+\mu_i}s\right)=\left(1-\sum_{i=1}^N\frac{\mu_i}{1+\mu_i}\right)s.
\end{equation}
The left-hand side of \eqref{eq:invert-G3} is a strictly decreasing function of $s$, bijective from $\R$ to $\R$. The right-hand side of \eqref{eq:invert-G3} is a non-decreasing linear function of $s$, since $\mu_i\in[0,
\frac{1}{N-1}]$ so that $\frac{\mu_i}{1+\mu_i}\in[0,\frac{1}{N}]$ and $\sum_{i=1}^N\frac{\mu_i}{1+\mu_i}\in[0,1]$. This means that \eqref{eq:invert-G3} has a unique solution $s\in\R$, which in turn uniquely determines $(y^1,...,y^N)\in\R^N$ through \eqref{eq:invert-G2}, completing the proof of invertibility of $G$ under the lower bounds on $RRA[U_i]$, $i\in\mathcal{N}$.

For a fixed $y=(y^1,...,y^N)\in\R^N$, the Jacobian of $G$ at $y$ is given by
\begin{equation}\label{eq:G-Jacobian}
JG(y)=\begin{pmatrix}
v_1 & -v_1\mu_1-\mu_1 & ... & -v_1\mu_1-\mu_1\\
-v_2\mu_2-\mu_2 & v_2 & ... & -v_2\mu_2-\mu_2\\
\vdots & \vdots & \ddots & \vdots \\
-v_N\mu_N-\mu_N & -v_N\mu_N-\mu_N & ... & v_N
\end{pmatrix},
\end{equation}
where $v_i=V'_i\left(y^i-\sum_{j\neq i}\mu_iy^j\right)\in\left(-\infty,-\varepsilon_i-\frac{\mu_i}{1+\mu_i}\right]$ for all $i\in\mathcal{N}$, and assume that either (i)~$\mu_i<\frac{1}{N-1}$ for some $i\in\mathcal{N}$, or (ii) $-v_i\leq R_i$ for all $y$ for some $i\in\mathcal{N}$. By Lemma \ref{lem:inverse-JG}, under either of these two conditions, $JG(y)$ is non-degenerate for all $y\in\R^N$ and the inverse $[JG(y)]^{-1}$ is bounded uniformly in $y\in\R^N$. Thus, by the inverse function theorem, $G^{-1}$ is continuously differentiable and $JG^{-1}(z)=[JG\vert_{G(z)}]^{-1}$ is bounded uniformly in $z\in\R^N$, i.e., $G^{-1}$ is globally Lipschitz.
\end{proof}

\begin{lem}\label{lem:inverse-JG}
Let $\mu_i\in[0,\frac{1}{N-1}]$, $\varepsilon_i>0$, and consider the following two situations:
\begin{enumerate}
\item For at least one $i\in\mathcal{N}$, $\mu_i<\frac{1}{N-1}$. We let $v_j$ be any number satisfying $-v_j\geq\varepsilon_j+\frac{\mu_j}{1+\mu_j}$ for all $j\in\mathcal{N}$.
\item For some $i\in\mathcal{N}$ we fix a constant $R_i\geq\varepsilon_i+\frac{\mu_i}{1+\mu_i}$ and let $v_i$ be any number in the range $\varepsilon_i+\frac{\mu_i}{1+\mu_i}\leq -v_i\leq R_i$. We let $-v_j\geq\varepsilon_j+\frac{\mu_j}{1+\mu_j}$ for $j\neq i$.
\end{enumerate}
Then, under either of the two conditions and with the corresponding ranges for $v_i$'s, the matrix
$$M=\begin{pmatrix}
v_1 & -v_1\mu_1-\mu_1 & ... & -v_1\mu_1-\mu_1\\
-v_2\mu_2-\mu_2 & v_2 & ... & -v_2\mu_2-\mu_2\\
\vdots & \vdots & \ddots & \vdots \\
-v_N\mu_N-\mu_N & -v_N\mu_N-\mu_N & ... & v_N
\end{pmatrix}$$
is invertible and $\norm{M^{-1}}$ is bounded uniformly for $v_i$'s lying in the specified ranges.
\end{lem}
\begin{proof}
The assertions of the lemma hold if and only if for every $k\in\mathcal{N}$, there exists a vector $u_k=(u_{1k},...,u_{Nk})\in\R^N$ such that $Mu_k=e_k$ and $u_k$ is bounded by a constant independent of $v_i$'s. Indeed, $M^{-1}=(u_1,...,u_N)$ in this case. We write down the $N$ equations for $u_k$:
\begin{equation}\label{eq:system1}
\begin{cases}
v_iu_{ik}-(v_i+1)\mu_i\sum_{j\neq i} u_{jk}=0,&\quad i\neq k,\\
v_ku_{kk}-(v_k+1)\mu_k\sum_{j\neq k}u_{jk}=1,&\quad i=k.
\end{cases}
\end{equation}
Denoting $s=\sum_{i\in\mathcal{N}}u_{ik}$, we can rewrite \eqref{eq:system1} as
\begin{equation}\label{eq:system2}
\begin{cases}
u_{ik}(v_i+\frac{\mu_i}{1+\mu_i})=(v_i+1)\frac{\mu_i}{1+\mu_i}s,&\quad i\neq k,\\
u_{kk}(v_k+\frac{\mu_k}{1+\mu_k})=\frac{1}{1+\mu_k}+(v_k+1)\frac{\mu_k}{1+\mu_k}s,&\quad i=k.
\end{cases}
\end{equation}
Dividing \eqref{eq:system2} by $v_i+\frac{\mu_i}{1+\mu_i}$ and summing over $i\in\mathcal{N}$, we obtain the equation for $s$:
\begin{equation}\label{eq:s}
s=\frac{1}{(1+\mu_k)(v_k+\frac{\mu_k}{1+\mu_k})}+\left(\sum_{i\in\mathcal{N}}\frac{(v_i+1)\frac{\mu_i}{1+\mu_i}}{v_i+\frac{\mu_i}{1+\mu_i}}\right)s.
\end{equation}
Since $\frac{(v_i+1)\frac{\mu_i}{1+\mu_i}}{v_i+\frac{\mu_i}{1+\mu_i}}=\frac{\mu_i}{1+\mu_i}+\frac{\mu_i}{(1+\mu_i)^2}\cdot\frac{1}{v_i+\frac{\mu_i}{1+\mu_i}}$, where
\begin{enumerate}
\item the first summand is at most $\frac{1}{N}$ and is equal to $\frac{1}{N}$ if and only if $\mu_i=\frac{1}{N-1}$;
\item the second summand is strictly negative and is bounded away from zero if $-v_i$ is bounded away from infinity (by any constant $R_i$);
\end{enumerate}
in either of the two scenarios from the statement of the lemma, we obtain that the expression
$$1-\sum_{i\in\mathcal{N}}\frac{(v_i+1)\frac{\mu_i}{1+\mu_i}}{v_i+\frac{\mu_i}{1+\mu_i}}>0$$
is bounded away from zero by a strictly positive constant independent of $v_i$'s. Furthermore, $\frac{1}{v_i+\frac{\mu_i}{1+\mu_i}}\in[-\frac{1}{\varepsilon_i},0)$ for all $i\in\mathcal{N}$. Hence, \eqref{eq:s}, or equivalently,
$$\left(1-\sum_{i\in\mathcal{N}}\frac{(v_i+1)\frac{\mu_i}{1+\mu_i}}{v_i+\frac{\mu_i}{1+\mu_i}}\right)s=\frac{1}{(1+\mu_k)(v_k+\frac{\mu_k}{1+\mu_k})}$$ has a unique solution $s$ bounded by a constant independent of $v_i$'s. Finally, expressing $u_{ik}$ from \eqref{eq:system2},
\begin{equation*}
\begin{cases}
u_{ik}=\frac{v_i+1}{v_i+\frac{\mu_i}{1+\mu_i}}\frac{\mu_i}{1+\mu_i}s=
\left(1+\frac{1}{(1+\mu_i)\left(v_i+\frac{\mu_i}{1+\mu_i}\right)}\right)\frac{\mu_i}{1+\mu_i}s,&\quad i\neq k,\\
u_{kk}=\frac{1}{(1+\mu_k)\left(v_k+\frac{\mu_k}{1+\mu_k}\right)}+\left(1+\frac{1}{(1+\mu_k)\left(v_k+\frac{\mu_k}{1+\mu_k}\right)}\right)\frac{\mu_k}{1+\mu_k}s,&\quad i=k,
\end{cases}
\end{equation*}
we find that $u_k=(u_{1k},...,u_{Nk})$ is uniquely determined and is bounded uniformly for $v_i$'s within the specified ranges.
\end{proof}

%%%%%%%%%%%%%%%%%%%%%%

\section{Proofs for Section \ref{section:inv-h}}\label{appendix:proofs}

The operator norm on $\R^{N\times N}$ associated with the $\infty$-norm on $\R^N$ is denoted by $\norm{\cdot}_\infty$. We recall that for a matrix $M=(m_{ij})_{i,j}\in\R^{N\times N}$, with $i$-th row denoted by $m_i$, we have
$$\norm{M}_{\infty}=\sup_{0\neq x\in\R^N}\frac{\norm{Mx}_\infty}{\norm{x}_\infty}=\max_{1\leq i\leq N}\sum_{j=1}^N\vert m_{ij}\vert=\max_{1\leq i\leq N}\norm{m_i}_1.$$

\begin{proof}[Proof of Proposition \ref{prop:close-to-CRRA}]

\textbf{\boldmath Step 1: Invertibility of $h$.} We need to show that $h$ is proper and $Jh$ is invertible if we take $\varepsilon$ small. By \eqref{eq:G-Jacobian}, \eqref{eq:J-CRRA}, and \eqref{eq:almost-CRRA}, we have
$$\norm{JG(y)-J_c}_\infty=\sup_{i\in\mathcal{N}}\left[1+(N-1)\mu_i\right]\cdot\vert v_i-r_i\vert\leq 2\varepsilon,\quad \text{for all }y\in\R^N.$$
By Lemma \ref{lem:inverse-est1}, for any $\delta>0$ there exists $\varepsilon>0$ such that $\norm{JG(y)-J_c}_\infty\leq 2\varepsilon$ implies that $JG(y)$ is invertible and $\norm{[JG(y)]^{-1}-J_c^{-1}}_\infty\leq\delta$. The equality \eqref{eq:h-Jacobian} states that $\nabla h^i(D)$ is the expectation of $\nabla [G^{-1}]^i(D+\ln Z\cdot \mathbbm{1})$ under the probability measure with relative density $\frac{\mathcal{H}^i(D)}{\Ex^\Q\left[\mathcal{H}^i(D)\right]}$ with respect to $\Q$. Since $\norm{\cdot}_1$ is convex, Jensen's inequality implies
$$\norm{\nabla h^i(D)-e_i^TJ_c^{-1}}_1\leq \sup_{\omega\in\Omega} \norm{\nabla [G^{-1}]^i(D+\ln Z\cdot \mathbbm{1})-e_i^TJ_c^{-1}}_1\leq \delta,\quad\text{for all }D\in\R^N,\ i\in\mathcal{N}.$$
Hence $\norm{Jh(D)-J_c^{-1}}_\infty=\max_{i\in\mathcal{N}}\norm{\nabla h^i(D)-e_i^TJ_c^{-1}}_1\leq\delta$ for all $D\in\R^N$. This means that the map $h$ is ``close to" the linear map $D\mapsto h(0)+ J_c^{-1}D$, therefore it should be invertible and proper. We formalize this statement next. By Lemma~\ref{lem:inverse-est1} again, $Jh$ is invertible if $\delta$ is small enough. By the mean value theorem, for every $D\in\R^N$ and $i\in\mathcal{N}$ there exists $\theta_i\in[0,1]$ such that $h^i(D)-h^i(0)=\nabla h^i(\theta_i D)D$, hence
$$\vert h^i(D)-h^i(0)-e_i^TJ_c^{-1}D\vert=\vert\left\{\nabla h^i(\theta_i D)-e_i^TJ_c^{-1} \right\}D\vert \leq \norm{\nabla h^i(\theta_i D)-e_i^TJ_c^{-1}}_1\norm{D}_\infty\leq \delta\norm{D}_\infty$$
and so $\norm{h(D)-h(0)-J_c^{-1}D}_\infty\leq \delta\norm{D}_\infty$. In particular, by the triangle inequality,
$$\norm{h(D)}_\infty\geq \norm{J_c^{-1}D}_\infty-\norm{h(0)}_\infty-\delta\norm{D}_\infty\geq \norm{J_c}^{-1}_\infty\norm{D}_\infty-\norm{h(0)}_\infty-\delta\norm{D}_\infty.$$
Clearly, for $\delta$ small enough the right-hand side tends to $\infty$ as $\norm{D}_\infty\to\infty$, and so does the left-hand side, i.e., $h$ is proper.

\textbf{Step 2: Convergence of optimal strategies.} Let us fix $(x_0^1,...,x_0^N)\in\R^N_+$ and write
\begin{align*}
(X^{1,*}_{T,c},...,X^{N,*}_{T,c})&=\mathcal{H}_c(D_c),\quad \text{where } D_c=h_c^{-1}\circ\ln(x^1_0,...,x^N_0),\\
(X^{1,*}_{T},...,X^{N,*}_{T})&=\mathcal{H}(D),\quad \text{where } D=h^{-1}\circ\ln(x^1_0,...,x^N_0).
\end{align*}
Without loss of generality, we can assume that for $V_i$, $i\in\mathcal{N}$, defined as in \eqref{eq:V-from-U}, $V_i(0)=0$. This amounts to normalizing the utility functions $U_i$ so that $U_i'(1)=1$, which does not affect the optimization problem \eqref{eq:util-max} and the corresponding optimal strategies. Hence, we can assume that $G(0)=0$.

\textit{Claim: $h^{-1}(z)\to h^{-1}_c(z)$ for every $z\in\R^N$ as $\varepsilon\downarrow0$ in \eqref{eq:almost-CRRA}.}

\textit{Proof of the Claim:} Since $G^{-1}(0)=G^{-1}_c(0)=0$ and $JG^{-1}(z)\to\ J^{-1}_c=JG^{-1}_c(z)$ uniformly in $z\in\R^N$, $G^{-1}\to G_c^{-1}$ pointwise. The definition \eqref{eq:def-of-hH} of $h$, pointwise convergence $G^{-1}\to G_c^{-1}$, and the fact that $G^{-1}$ and $G_c^{-1}$ are globally Lipschitz and all moments of $Z$ are finite imply that $h(0)\to h_c(0)$ as $\varepsilon\downarrow 0$ by the dominated convergence theorem. By linearity of $h_c^{-1}$, this implies $h^{-1}_{c}(h(0))\to 0$ as $\varepsilon\downarrow 0$. For an arbitrary $z\in\R^N$, by the mean value theorem applied (coordinate-wise) to function $h^{-1}-h_c^{-1}$, we can write
\begin{align*}
\norm{h^{-1}(z)- h^{-1}_c(z)}&\leq \norm{[h^{-1}(z)-h^{-1}_c(z)]-[h^{-1}(h(0))-h_c^{-1}(h(0))]}+\norm{h_c^{-1}(h(0))}\\
&\leq \left(\sup_{\R^N}\norm{Jh^{-1}-J_c}\right)\cdot \norm{z-h(0)}+\norm{h_c^{-1}(h(0))}\to 0\quad \text{as }\varepsilon\downarrow 0,
\end{align*}
where the uniform convergence $\norm{Jh^{-1}-J_c}\to 0$ again follows by Lemma \ref{lem:inverse-est1} and by the uniform convergence $\norm{Jh-J_c^{-1}}\to 0$ shown before. $\qed$

The above claim shows that $D\to D_c$ as $\varepsilon\downarrow 0$. By Lemma \ref{lem:continuity}, $\mathcal{H}(D)\to\mathcal{H}(D_c)$ almost surely and in $L^p$. Therefore, in order to show the convergence $\mathcal{H}(D)\to\mathcal{H}_c(D_c)$, it remains to show that $\mathcal{H}(D_c)\to\mathcal{H}_c(D_c)$ almost surely and in $L^p$ as $\varepsilon\downarrow 0$. Almost sure convergence follows by pointwise convergence $G^{-1}\to G^{-1}_c$. Convegence in $L^p$ for all $p>0$ follows by pointwise convergence and by the dominated convergence theorem due to uniform bounds on $\norm{JG^{-1}}$ and finiteness of all moments of $Z$.
\end{proof}

\begin{proof}[Proof of Proposition \ref{prop:lambda-small}]
\textbf{\boldmath Step 1: Invertibility of $h$.} Let $r_i=\varepsilon_i+\frac{\mu_i}{1+\mu_i}$, where $\varepsilon_i>0$ is given as in Assumption \ref{ass:utility2} for the starting value of $\lambda_i$, so that $R_i\geq RRA[U_i](x)\geq r_i$ for all $x>0$. In the proof below, we choose $\varepsilon>0$ small enough so that the condition
\begin{equation}\label{eq:mu-small}
\mu_i\cdot \max(|R_i-1|,|1-r_i|)\leq \varepsilon,\quad \text{for all }i\in\mathcal{N},
\end{equation}
implies that $h$ is proper and has a non-degenerate Jacobian at every point. Condition \eqref{eq:mu-small} ensures that all non-diagonal entries of the Jacobian $JG(y)$ in \eqref{eq:G-Jacobian} are bounded by $\varepsilon$, hence
\begin{equation}\label{eq:eps-delta}
\norm{JG(y)-JG_n(y)}_\infty\leq (N-1)\varepsilon,\quad \text{for all }y\in\R^N,
\end{equation}
where $G_n:\R^N\to\R^N$ is defined as in \eqref{eq:no-comp-functions}, so $JG_n(y)$ is a diagonal matrix with diagonal entries equal to those of $JG(y)$. By Lemma~\ref{lem:inverse-est1}, for every $\delta>0$ there exists $\varepsilon>0$ such that \eqref{eq:eps-delta} implies
\begin{equation}\label{eq:decompose-JG}
\norm{JG^{-1}(z)-JG_n^{-1}\vert_{y=G_n\circ G^{-1}(z)}}_\infty\leq \delta,\quad \text{for all }z\in\R^N.
\end{equation}
For every $z\in\R^N$, we have $[JG_n\vert_{y=G_n\circ G^{-1}(z)}]_{ii}=v_i\leq R_i$ with $v_i$ defined as in \eqref{eq:G-Jacobian}, hence $[JG_n^{-1}\vert_{y=G_n\circ G^{-1}(z)}]_{ii}\geq \frac{1}{R_i}>0$ is bounded away from zero uniformly in $z\in\R^n$ for every $i\in\mathcal{N}$. By decomposing for each $z\in\R^N$ the Jacobian $JG^{-1}(z)$ in this way -- into a diagonal matrix with the $i$'s diagonal entry uniformly bounded away from zero by $\frac{1}{R_i}$ for every $i\in\mathcal{N}$ and a matrix with $\norm{\cdot}_\infty$-norm at most $\delta$, -- and recalling the formula \eqref{eq:h-Jacobian} for the Jacobian  $Jh(D)$ as a ``weighted average" of $JG^{-1}(D+\ln Z\cdot\mathbbm{1})$, we conclude that the same property holds for $Jh(D)$ for every $D\in\R^N$: there exists a diagonal matrix $M(D)$ with $M_{ii}(D)\geq \frac{1}{R_i}$ for $i\in\mathcal{N}$ and such that $\norm{Jh(D)-M(D)}_\infty\leq\delta$. By Lemma~\ref{lem:inverse-est2}, for $\delta>0$ sufficiently small (and independent of $D$) this implies invertiblity of $Jh(D)$ and a uniform upper bound on $\norm{(Jh(D))^{-1}}_\infty$. By the mean value theorem, for every $D\in\R^N$ and $i\in\mathcal{N}$ there exists $\theta_i\in[0,1]$ such that $h^i(D)-h^i(0)=\nabla h^i(\theta_i D)D$. Hence $h(D)-h(0)=\mathcal{J}(D)D$, where the matrix $\mathcal{J}(D)\in\R^{N\times N}$ is defined so that its $i$-th row is equal to $\nabla h^i(\theta_i D)$. In particular, there exists a diagonal matrix $\mathcal{M}(D)\in\R^{N\times N}$ such that $\mathcal{M}_{ii}(D)\geq \frac{1}{R_i}$ for $i\in\mathcal{N}$ and $\norm{\mathcal{J}(D)-\mathcal{M}(D)}_\infty\leq\delta$. We have
$$\norm{h(D)}_\infty\geq \norm{\mathcal{J}(D)D}_\infty-\norm{h(0)}_\infty\geq \min_{i\in\mathcal{N}}{\frac{1}{R_i}}\cdot\norm{D}_\infty-\norm{h(0)}_\infty-\delta\norm{D}_\infty,$$ 
and therefore $h$ is proper for $\delta$ small enough.

\textbf{Step 2: Convergence of optimal strategies.} Let us fix $(x_0^1,...,x_0^N)\in\R^N_+$ and write
\begin{align*}
(X^{1,*}_{T,n},...,X^{N,*}_{T,n})&=\mathcal{H}_n(D_n),\quad \text{where } D_n=h_n^{-1}\circ\ln(x^1_0,...,x^N_0),\\
(X^{1,*}_{T},...,X^{N,*}_{T})&=\mathcal{H}(D),\quad \text{where } D=h^{-1}\circ\ln(x^1_0,...,x^N_0).
\end{align*}
We assume that $G(0)=G_n(0)=0$ by normalizing the utility functions so that $U_i'(1)=1$ for all $i$.

\textit{Claim: Convergence $G^{-1}(z)\to G^{-1}_n(z)$, $h(z)\to h_n(z)$, and $h^{-1}(z)\to h^{-1}_n(z)$ holds for every $z\in\R^N$ as $\varepsilon\downarrow0$ in $0\leq\lambda_i\leq \varepsilon$, $i\in\mathcal{N}$.}

\textit{Proof of the Claim:} Clearly $G\to G_n$ pointwise as $\varepsilon\downarrow 0$ by \eqref{eq:express-G}. The map $G^{-1}$ is globally Lipschitz by Proposition \ref{prop:G-invertible} and, as can be seen from the proof of Lemma \ref{lem:inverse-JG}, the Lipschitz constant is uniform as $\varepsilon\downarrow 0$. Hence, for every $z\in\R^N$, we obtain
$$\norm{G^{-1}(z)-G^{-1}_n(z)}=\norm{G^{-1}(z)-G^{-1}(G\circ G^{-1}_n(z))}\leq L\norm{G_n(G^{-1}_n(z))-G(G^{-1}_n(z))}\to 0$$
as $\varepsilon\downarrow0$. By the definition \eqref{eq:def-of-hH} of $h$, since $G^{-1}\to G^{-1}_n$, since the maps $G^{-1}$, $G^{-1}_n$ are globally Lipschitz with uniform in $\varepsilon$ Lipschitz constant, and since all moments of $Z$ are finite, dominated convergence theorem implies $h(z)\to h_n(z)$ as $\varepsilon\downarrow 0$ for every $z\in\R^N$. For $h^{-1}\to h^{-1}_n$ we conclude similarly as with $G^{-1}\to G^{-1}_n$ above, noting that Step 1 provides a uniform bound on $\norm{Jh^{-1}}_\infty$ for all $\varepsilon$ small enough, i.e., $h^{-1}$ is globally Lipschitz with a uniform in $\varepsilon$ Lipschitz constant. $\qed$

The above claim shows that $D\to D_n$ as $\varepsilon\downarrow 0$. Convergence $\mathcal{H}(D)\to\mathcal{H}_n(D_n)$ almost surely and in $L^p$ follows from convergence $G^{-1}\to G^{-1}_n$ same as in the proof of Proposition \ref{prop:close-to-CRRA}.
\end{proof}

In the following lemmata, we list for reference several standard linear algebra facts that are used in the proofs above.

\begin{lem}\label{lem:inverse-est1}
Let $\norm{\cdot}$ be an operator norm on $\R^{N\times N}$ generated by a norm $\norm{\cdot}$ on $\R^N$. Let $S\in\R^{N\times N}$ be an invertible matrix. If $T\in\R^{N\times N}$ satisfies $\norm{S-T}\leq\varepsilon\norm{S^{-1}}^{-1}$, where $0<\varepsilon<1$, then $T$ is also invertible and $\norm{T^{-1}}\leq\frac{1}{1-\varepsilon}\norm{S^{-1}}$. Moreover, in this case, $\norm{S^{-1}-T^{-1}}\leq \frac{\varepsilon}{1-\varepsilon}\norm{S^{-1}}$.
\end{lem}

\begin{proof}
Note that $\norm{Id-S^{-1}T}\leq\norm{S^{-1}}\cdot\norm{S-T}<\varepsilon<1$. Using the Neumann series, we conclude that the matrix $S^{-1}T$ is invertible and the inverse is given by $\left(Id-(Id-S^{-1}T)\right)^{-1}=\sum_{k\geq 0}\left(Id-S^{-1}T\right)^k$. In particular, this relation implies $\norm{T^{-1}S}\leq \sum_{k\geq 0}\varepsilon^k=\frac{1}{1-\varepsilon}$, hence $\norm{T^{-1}}\leq\norm{T^{-1}S}\cdot\norm{S^{-1}}\leq\frac{1}{1-\varepsilon}\norm{S^{-1}}$. The last assertion follows from $\norm{S^{-1}-T^{-1}}\leq\norm{S^{-1}}\cdot\norm{S-T}\cdot\norm{T^{-1}}\leq \frac{\varepsilon}{1-\varepsilon}\norm{S^{-1}}$.
\end{proof}

\begin{lem}\label{lem:inverse-est2}
Let $M=(m_{ij})_{i,j}\in\R^{N\times N}$ be a strictly diagonally dominant matrix with positive diagonal entries:
$$m_{ii}-\sum_{j\neq i}\vert m_{ij}\vert\geq \varepsilon >0\quad \text{for all}\quad i=1,...,N.$$
Then $M$ is invertible and $\norm{M^{-1}}_\infty\leq \frac{1}{\varepsilon}$.
\end{lem}

\begin{proof}
Let $m>\max_{i}m_{ii}$ and define $S=Id-\frac{M}{m}$. Then
$$\norm{S}_\infty=\max_i \left[\left(1-\frac{m_{ii}}{m}\right)+\sum_{j\neq i}\frac{\vert m_{ij}\vert}{m}\right]
\leq 1-\frac{\varepsilon}{m}<1.$$
Using the Neumann series expansion, we obtain that $\frac{M}{m}$ is invertible and $\left(\frac{M}{m}\right)^{-1}=(Id-S)^{-1}=\sum_{k\geq 0}S^k$.
Hence, $M$ is invertible and
$m\norm{M^{-1}}_\infty\leq \sum_{k\geq 0}\norm{S}_\infty^k\leq \sum_{k\geq 0}\left(1-\frac{\varepsilon}{m}\right)^k=\frac{m}{\varepsilon}.$
\end{proof}

\end{appendices}

\textbf{Acknowledgements.} I would like to thank Mihai S\^irbu for suggesting the research direction of relative performance criteria and reading an earlier draft of the paper and Gordan \v Zitkovi\'c for his comments.

\textbf{Funding.} Partially supported by the National Science Foundation under Grant DMS-1908903 and by the University Graduate Continuing Fellowship of the University of Texas at Austin.

\footnotesize{\bibliography{bibliography}}

\end{document}